%% file: main.tex
\newif\ifextended
\extendedtrue

\newif\ifbadges
\badgestrue

\documentclass[letterpaper,twocolumn,10pt]{article}
\ifextended
\usepackage{usenix2020-for-arxiv}
\else
\usepackage{usenix2020}
\fi

\usepackage{amsthm}
\newtheorem{theorem}{Theorem}[section]

\newtheorem{lemma}[theorem]{Lemma}

\newif\ifsilent
\silentfalse
\ifsilent
\newcommand{\anote}[1]{}
\newcommand{\ynote}[1]{}
\newcommand{\mnote}[1]{}
\newenvironment{revision}{}{}
\else
\newcommand{\anote}[1]{\textcolor{red}{Amit: #1}}
\newcommand{\ynote}[1]{\textcolor{blue}{Yossi: #1}}
\newcommand{\mnote}[1]{\textcolor{magenta}{Moshe: #1}}
\newenvironment{revision}{\color{cyan}}{}
\fi

\newcommand{\ignore}[1]{}

\usepackage{titlesec}
\titlespacing{\paragraph}{%
  0pt}{
  0.5\baselineskip}{
  1em}

\usepackage{tikz}
\usepackage{amsmath}
\usepackage{threeparttable,booktabs}

\usepackage[utf8]{inputenc}
\usepackage{bbding}
\usepackage{pifont}
\newcommand{\xmark}{\ding{55}}
\usepackage{hyperref}
\usepackage{amssymb}
\usepackage{nicefrac}
\usepackage{mathtools}
\usepackage{algorithm}
\usepackage{algcompatible}
\usepackage{tabularx}
\usepackage{setspace}

\usepackage{graphicx}

\usepackage{comment}
\usepackage[noend]{algpseudocode}

\usepackage[normalem]{ulem}
\usepackage{units}
\usepackage{verbatim}
\usepackage{appendix}

\usepackage[super]{nth}
\usepackage{lscape}
\usepackage{multirow}

\usepackage{pdfpages}
\usepackage{float}
\usepackage{pgfplots}
\ifextended
\pgfplotsset{compat=1.17}
\else
\pgfplotsset{compat=1.18}
\fi

\usepackage{xurl}
\usepackage{balance}

\newcommand*\circled[1]{\protect\tikz[baseline=(char.base)]{
            \protect\node[shape=circle,draw,inner sep=1pt] (char) {#1};}}

\newcommand{\exttext}[2]{\ifextended{#2}\else{#1}\fi}
\newcommand{\extref}[1]{\exttext{\cite[\S A]{tcp-sport-tracking-extended}}{\autoref{#1}}}

\ifextended
\else
\usepackage{nopageno}  
\fi

\ifbadges
\usepackage[available,functional,reproduced]{usenixbadges}
\fi

\begin{document}

\date{}

\title{Device Tracking via Linux's New TCP Source Port Selection Algorithm
\ifextended
\\
(Extended Version)
\thanks{This is an extended version of a paper with the same name that will be presented in the \nth{32} Usenix Security Symposium (USENIX 2023).}
\fi
}

\author{
{\rm Moshe Kol}
\\
Hebrew University of Jerusalem
 \and
 {\rm Amit Klein}
\\
Hebrew University of Jerusalem
 \and
 {\rm Yossi Gilad}
 \\
Hebrew University of Jerusalem
}

\date{\vspace{-\baselineskip}}
\maketitle

\begin{abstract}
\input{abstract}
\end{abstract}

\widowpenalty=1000
\clubpenalty=1000

\input{sec-intro}
\input{sec-related-work}

\input{sec-background}

\input{sec-device-tracking} 
\input{sec-implementation}

\input{sec-evaluation}
\input{sec-recommendations}

\input{sec-conclusions}
\input{sec-vendor-status}

\section{Availability}\label{sec:availability}

Our code is publicly available online, along with instructions for reproducing
our results, in \url{https://github.com/0xkol/rfc6056-device-tracker}.

\section*{Acknowledgements}
We thank the anonymous Usenix 2022 reviewers for their thoughtful comments, suggestions and feedback.  This work was supported in part by the Hebrew University cyber security research center. Yossi Gilad was partially supported by the Alon fellowship.

\newpage
\bibliographystyle{plain}

\IfFileExists{./main.bbl}
{

\input{./main.bbl}
}
{
    \bibliography{akbib8,delta9}
}

\newpage
\appendix
\appendixpage

\exttext{}{\input{app-analysis}}

\input{sec-other-use-cases}
\input{app-algorithm5}
\input{app-cloud-experiment}

\end{document}

%% file: abstract.tex
We describe a tracking technique for Linux devices, exploiting a new TCP source port generation mechanism recently introduced to the Linux kernel. This mechanism is based on an algorithm, standardized in RFC 6056, for boosting security by better randomizing port selection. Our technique detects collisions in a hash function used in the said algorithm, based on sampling TCP source ports generated in an attacker-prescribed manner. These hash collisions depend solely on a per-device key, and thus the set of collisions forms a device ID that allows tracking devices across browsers, browser privacy modes, containers, and IPv4/IPv6 networks (including some VPNs). It can distinguish among devices with identical hardware and software, and lasts until the device restarts. 

We implemented this technique and then tested it using tracking servers in two different locations and with Linux devices on various networks. We also tested it on an Android device that we patched to introduce the new port selection algorithm. 
The tracking technique works in real-life conditions, and we report detailed findings about it, including its dwell time, scalability, and success rate in different network types. We worked with the Linux kernel team to mitigate the exploit, resulting in a security patch introduced in May 2022 to the Linux kernel, and we provide recommendations for better securing the port selection algorithm in the paper. 

%% file: sec-intro.tex
\section{Introduction}
Online browser-based device tracking is a widespread practice, employed by many Internet websites and advertisers. It allows identifying users across multiple sessions and websites on the Internet. 
A list of motivations for web-based device tracking (fingerprinting) is listed in \cite{Acar:2013:FDW:2541806.2516674}  and includes ``fraud detection, protection against account hijacking, anti-bot and anti-scraping services, enterprise security management, protection against DDOS attacks, real-time targeted marketing, campaign measurement, reaching customers across devices, and limiting the number of accesses to services''.

Device tracking is often performed to personalize ads or for surveillance purposes. It can either be done by sites that users visit or by third-party companies (e.g. advertisement networks) which track users across multiple websites and applications (``cross-site tracking''). Traditionally, cross-site tracking was implemented via \nth{3} party cookies. However, nowadays, users are more aware of the cookies' privacy hazards, and so they use multiple browsers, browser privacy mode, and cookie deletion to avoid such tracking. In addition, support for \nth{3} party cookies in major browsers is being withdrawn due to privacy concerns~\cite{safari-3rd-party-cookies,chrome-3rd-party-cookies}. Trackers are, therefore, on the look for new tracking technologies, particularly ones that can work across sites and across browsers and privacy modes, thereby breaking the isolation the latter attempt to provide. 

Probably the most alarming impact of device tracking is the degradation of user privacy -- when a user's device can be tracked across network changes, different browsers, VPNs, and browser privacy modes. This means that users who browse to one site with some identity (e.g., user account), then browse to another site, from another browser, another network (or VPN), and perhaps at another time altogether, using a completely different and unrelated second identity, may still have the two identities linked. 

Often, device tracking techniques are used in a clandestine manner, without the user's awareness and without obtaining the user's explicit consent. 
This motivates researchers to understand the challenges of device tracking, find new tracking techniques that can be used without consent, and work with the relevant software vendors to eliminate such techniques and raise awareness of these new kinds of attacks.

In this paper, we present a new browser-based tracking technique that 
supports tracking across IPv4 and IPv6 networks, browsers, VPNs, and browser privacy modes. 
Our tracking technique can provide up to 128 bits of entropy for the device ID (in the Linux implementation) and requires negligible CPU and RAM resources for its operation. 
Our technique uses standard web technologies such as Javascript, WebRTC TURN (in Chrome), and XHR (in Firefox). It assumes that a browser renders a web page with an embedded tracking HTML snippet that communicates with a \nth{1}-party tracking server (i.e., there is no reliance on common infrastructure among the tracking websites). The tracking server then calculates a device ID. This ID is based on kernel data. Therefore, the same device ID is calculated by any site that runs the same logic, regardless of the network from which the tracked device arrives, or the browser used. 

The tracking technique is based on observing the TCP source port numbers generated by the device's TCP/IP stack, which is implemented in the operating system kernel. There are several popular TCP source port generation algorithms that differ in the level of security vs. functionality they deliver. Security-wise, TCP source ports should be as unpredictable as possible to off-path attackers~\cite[\S1]{rfc6056}. Functionality-wise, TCP source ports should not repeat too often (to mitigate the ``instance-id collision'' problem~\cite[\S2.3]{rfc6056}). 
 
 RFC 6056 ``Recommendations for Transport-Protocol Port Randomization'' \cite[\S 3.3]{rfc6056} lists five algorithms used by different operating systems to generate TCP source port numbers. According to RFC 6056, the ``Double-Hash Port Selection'' algorithm \cite[\S 3.3.4]{rfc6056} offers the best trade-off between the design goals of TCP source ports (see \autoref{sec:rfc6056}),
and indeed it was recently adopted with minor modifications by Linux (starting with kernel version 5.12-rc1).  Our analysis targets this port selection algorithm, which we expect to propagate into Android devices as well (as Android 13 launches with kernel version 5.15 \cite{android-kernels}). 
 
Our technique finds hash collisions in one of the algorithm's hash functions. 
These collisions depend only on a secret hashing key that the OS kernel creates on boot time and maintains until the system is shut down. Thus, the set of collisions forms a device ID that spans the lifetime of this key, surviving changes to networks, transitions into and out of sleep mode, and using containers on the same machine. It does not rely on the specific choice of the hash functions for the RFC 6056 algorithm beyond the RFC's own requirement, and as such, it presents a generic attack against the RFC algorithm. 
Since our technique relies on the client's port selection algorithm, it also has some limitations. Specifically, it is ineffective when the client uses Tor or an HTTP forward proxy since they terminate the TCP connection originating at the device and establish their own TCP connection with the tracking server. Furthermore, if a middlebox rewrites the TCP source ports or throttles the traffic, it can interfere with our technique. However, we note that this kind of interference is typically NAT-related and as such is unlikely to apply to IPv6 networks.

We implemented the device tracking technique and tested it with Linux devices across various networks (cellular, WiFi, Ethernet), VPNs, browsers (Chrome, Firefox), browser privacy modes, and Linux containers. It requires the browser to dwell on the web page for 10 seconds on average, which aligns with our theoretical analysis. The resources device tracking requires from the attacker are low, which allows calculating IDs for millions of devices per tracking server.
Since off-the-shelf Android devices have not yet deployed Linux kernels that use the new port selection algorithm, we introduced the new algorithm through a patch to a Samsung Galaxy S21 mobile phone (Android device) and tested it.
We recommended countermeasures to the Linux kernel team and worked with them on a security patch that was recently released (May 2022). We discuss these recommendations in the paper.

In sum, we make the following contributions:
{\begin{itemize}
  \setlength{\itemsep}{4pt}
  \setlength{\parskip}{0pt}
  \setlength{\parsep}{0pt}
    \item Analysis of RFC 6056's ``Double-Hash Port Selection'' algorithm, showing that a practical device tracking attack can be mounted against devices using this algorithm.
    \item Adaptation of our device tracking technique to Linux.
    \item Demonstration and measurements of our device tracking technique across the Internet, in practical settings, under various conditions (browsers, networks, devices, containers, VPNs).
    \item Full source code of our demonstration tracking server.
\end{itemize}
}

%% file: sec-related-work.tex
\section{Related Work}
The challenges facing device tracking nowadays revolve around the reliability and scope of the available tracking techniques. Ideally, a web-based tracking technique should work across browser privacy modes (switching between the normal browsing mode and the privacy mode such as ``incognito''), across browsers, across networks (switching between cellular networks, WiFi and Ethernet networks), and even across VPN connections. Furthermore, a tracking technique should address the ``golden image'' challenge \cite{KP18}, wherein an organization provisions a multitude of identical devices (fixed hardware and software configuration) to its employees. The tracking technique should thus tell these devices apart, even though there is no difference in their hardware and software. 

Device tracking techniques can be categorized into {\em tagging} and {\em fingerprinting} techniques~\cite{Wramner}. Tagging techniques insert an ID to the device, typically at the browser level (e.g., a resource cached by the browser or an object added to the standard browser storage such as cookies or {\tt localStorage}). Fingerprinting techniques measure a system or browser's features that can tell devices apart, such as fonts, hardware, and system language. 

Klein and Pinkas~\cite{KP18} provide an extensive review of browser-based tracking techniques, current for 2019. They evaluate the techniques' coverage of the golden image challenge 
and ability to cross the privacy mode gap. They find that typically, fingerprinting techniques fail to address the golden image challenge, while tagging techniques fail when users use the browser in privacy mode.
They found that no single existing technique was able to fulfill both requirements. The technique suggested in \cite{KP18} does not work across networks, and as such, its practicality is limited. A recent (2020) list of fingerprinting techniques (none of which overcomes the golden image challenge by default) is provided in \cite{10.1145/3386040}.

Since the analysis provided in \cite{KP18}, the browser-based tracking landscape grew with several new techniques. A tagging technique~\cite{DBLP:conf/ndss/SolomosKKP21} used the browser favicons cache to store a device ID. 
This was since fixed in Chrome 91.0.4452.0 \cite{chrome-bug-1096660,chrome-commit-c1ce99f4ed592d25d19ab37c18d6d9512655b14a}.
 A fingerprinting technique based on measuring GPU features from the WebGL 2.0 API is provided in \cite{Jiang2020FingerprintingWA}. A fingerprinting method based on TLS client features and behavior is described in \cite{TLSfp}. All the above techniques suffer from their respective category (fingerprinting/tagging) drawbacks.

A technique somewhat similar to \cite{KP18} that uses the stub resolver DNS cache and times DNS cache miss vs. DNS cache hit events is presented in \cite{8937592}, but this technique (like \cite{KP18}) does not work across networks. 

The ``Drawn Apart'' browser-based tracking technique~\cite{Drawn-Apart} is based on measuring timing deviations in GPU execution units. This technique is oblivious to the device's network configuration and the choice of browser and privacy mode and can also tell apart devices with identical hardware and software configurations. However, it does so with limited accuracy -- 36.7\%-92.7\% in lab conditions~\cite[Table 1]{Drawn-Apart}, which is insufficient for large-scale tracking. 

Klein et al.'s works \cite{KP19-usenix,flowlabel,Portland} revolved around a tracking concept based on kernel data leakage, which identifies individual devices. The leakage occurred in various network protocol headers (IPv4 ID, IPv6 flow label, UDP source port). All of them were quickly fixed by the respective operating system vendors due to their severity and impact and were no longer in effect when our research was conducted.

%% file: sec-background.tex
\section{Background}\label{sec:rfc6056}
 RFC 6056 \cite[Section 3.3]{rfc6056} analyzes five TCP source port allocation algorithms and defines several design goals~\cite[Section 3.1]{rfc6056}. The two goals relevant to this work are:
 \begin{itemize}
   \setlength{\itemsep}{4pt}
  \setlength{\parskip}{0pt}
  \setlength{\parsep}{0pt}
     \item Minimizing the predictability of the ephemeral port numbers used
      for future transport-protocol instances.
      \item Minimizing collisions of instance-ids [TCP 4-tuples].
 \end{itemize}
 The first goal aims for security against blind TCP attacks, such as blind reset or data injection attacks \cite{rfc5961}. The second goal is functionality-related and ensures that successive port assignments for the same destination (IP and port) do not re-use the same source port since this can cause a TCP failure at the remote end. That is because if the device terminates one TCP connection, the server may still keep it active (in the TCP TIME\_WAIT state) while the device attempts to establish a second connection with the same TCP 4-tuple, which will fail since the server has this 4-tuple still in use.
 
\paragraph{Double-Hash Port Selection Algorithm.}
Our focus is on RFC 6056's Algorithm 4, ``Double-Hash Port Selection Algorithm'' (DHPS), which is detailed in~\autoref{alg:alg-4}. This algorithm selects a TCP source port for a given $\mathit{IP}_\mathit{SRC},\mathit{IP}_\mathit{DST},\mathit{PORT}_\mathit{DST}$, which we term the connection's {\em 3-tuple}. Thus the algorithm completes a 3-tuple into a (TCP) 4-tuple. In this algorithm, $\mathit{table}$ is a ``perturbation table'' of $T$ integer counters (in the Linux kernel, $T=256$), $F$ is a cryptographic keyed-hash function which maps its inputs to a large range of integers, e.g., $[0,2^{32}-1]$, and $G$ is a cryptographic keyed-hash function which maps its inputs to $[0,T-1]$. The TCP source ports the algorithm produces are in the range $[\mathit{min\_ephemeral},\mathit{max\_ephemeral}]$ (in the Linux kernel, by default, $\mathit{min\_ephemeral}=32768, \mathit{max\_ephemeral}=60999$). DHPS calculates an index $i$ to a counter in $\mathit{table}$ based on a keyed-hash ($G$) of the given 3-tuple and uses the counter value offset by another keyed-hash ($F$) of the 3-tuple as a first candidate for a port number (using modular arithmetic to produce a value in $[\mathit{min\_ephemeral},\mathit{max\_ephemeral}]$). DHPS then checks whether the port number candidate is suitable (\textsc{CheckSuitablePort}). This check is intentionally under-specified in the RFC so that each implementation may run its own logic. For example, the Linux kernel checks whether there already is a 4-tuple with these parameters. If this check passes, DHPS returns the candidate port; otherwise, it increments the candidate port and runs the check again.

\begin{algorithm}[t]
\caption{DHPS Source Port Selection (RFC 6056 \S3.3.4)}
\label{alg:alg-4}
\begin{algorithmic}[1]
\Procedure{SelectEphemeralPort}{}
\State $\mathit{num\_ephemeral} \gets $
\State \hskip2.0em $\mathit{max\_ephemeral} - \mathit{min\_ephemeral} + 1$
\State $\mathit{offset} \gets F_{K_1}(\mathit{IP}_\mathit{SRC},\mathit{IP}_\mathit{DST},\mathit{PORT}_\mathit{DST})$
\State $\mathit{index} \gets G_{K_2}(\mathit{IP}_\mathit{SRC},\mathit{IP}_\mathit{DST},\mathit{PORT}_\mathit{DST})$ 
\State $\mathit{count} \gets \mathit{num\_ephemeral}$
\Repeat
\State $\mathit{port} \gets \mathit{min\_ephemeral}+$
\State \hskip2.0em $((\mathit{offset}+\mathit{table}_\mathit{index}) \mod \mathit{num\_ephemeral})$
\State $\mathit{table}_\mathit{index} \gets \mathit{table}_\mathit{index}+1$
\If{\Call{CheckSuitablePort}{$\mathit{port}$}}
\State return $\mathit{port}$
\EndIf
\State $\mathit{count} \gets \mathit{count}-1$
\Until{$\mathit{count}=0$}
\State return ERROR
\EndProcedure
\end{algorithmic}
\end{algorithm}

\paragraph{Port selection in the Linux kernel.}
Linux version 5.12-rc1 switched from RFC 6056's Algorithm 3 (``Simple Hash-Based Port Selection Algorithm'') to DHPS,  quoting security and privacy concerns as the reason for this change \cite{Linux-patch-RFC6056}. Starting from this version, the Linux kernel uses DHPS to generate TCP source ports for outbound TCP connections over IPv4 and IPv6. The Linux implementation and its few minor modifications are discussed in \autoref{sec:linux}. 
Linux kernel version 5.15 is the first long-term service (LTS) kernel version in which DHPS is used, thus LTS Linux installations using unpatched kernel 5.15 and above are vulnerable (see \autoref{sec:vendor-status}). 

\paragraph{Port selection in Android.}
 The Android operating system kernel is based on Linux and as such vulnerabilities found in Linux may also impact Android. The TCP source port selection algorithm in Android depends on the underlying Linux kernel version: devices running Linux kernels $\le$ 5.10 do not use DHPS and therefore are not vulnerable, whereas devices running a more recent kernel use DHPS and {\em may} be vulnerable if unpatched. At the time of writing, Android devices on the market use kernel version 5.10, even for Android 13, though Android 13 running kernel 5.15 is likely to be released in the near future \cite{android-kernels}. To assess the feasibility of our attack on Android, we conducted an experiment using an Android device with a modified kernel that includes the flaw (see \autoref{sec:android} for results). Since the vulnerability was patched on Linux in May 2022, and the patch was merged into Android as well, we expect future Linux and Android devices that use DHPS to be safe.

%% file: sec-device-tracking.tex
\section{Device Tracking Based on DHPS}

\paragraph{Attack model.} We assume the victim's (Linux-based) device runs a browser that renders a web page containing a ``tracking snippet'' -- a small piece of HTML and Javascript code (which runs in the browser's Javascript sandbox). The snippet implements the client-side logic of the tracking technique. It can be embedded in web pages served by, e.g., commerce websites or \nth{3}-party advertisements. When a device visits these pages, 
the tracking server logic calculates a unique ID for that device, allowing tracking it both with respect to the time dimension and the space dimension (visited websites).

\paragraph{Device ID.} When the browser renders the tracking snippet and executes the Javascript code in it, the code makes the browser engage in a series of TCP connection attempts with the attacker's tracking server, interleaved with TCP connection attempts to a localhost address. 
By observing the browser's traffic TCP source ports at the tracking server, the attacker can deduce hash collisions on $G_{K_2}$ (this is explained later). The attack concludes when the attacker collects enough pairs of hash collisions on $G_{K_2}$ where $\mathit{IP}_\mathit{SRC},\mathit{IP}_\mathit{DST}$ are fixed loopback addresses, i.e., pairs $(x,y)$ such that
\begin{align}
G_{K_2}(\mathit{IP}_\mathit{SRC}=127.0.0.1,\mathit{IP}_\mathit{DST}=127.1.2.3,\mathit{PORT_{DST}}=x)=\nonumber\\
G_{K_2}(\mathit{IP}_\mathit{SRC}=127.0.0.1,\mathit{IP}_\mathit{DST}=127.1.2.3,\mathit{PORT_{DST}}=y)\nonumber
\end{align}
These pairs depend only on $K_2$, and as such, they represent information on $K_2$ and on it alone. $K_2$ is statistically unique per device (up to the birthday paradox) and does not rely on the current browser, network, or container. Thus, the set of collisions $\{(x_i,y_i)\}$ forms a device ID that persists across browsers, browser privacy mode, network switching, some VPNs, and even across Linux containers. The ID is invalidated only when the device reboots or shuts down.

The use of loopback addresses is critical to the effectiveness of the attack. Using the Internet-facing address of the device does not yield a consistent device ID across networks. While hash collisions based on the Internet-facing address {\em can} be calculated, they are of no use when the device moves across networks because the device typically obtains a new, different Internet-facing address whenever it connects to another network. Thus, the collisions calculated for different networks will likely be completely different sets.

\paragraph{Limitations.} Our technique tracks client devices through their source port choice. A middlebox such as a NAT may modify the client's source port selection and cause the tracking service to fail to compute a consistent device ID. Furthermore, a device establishing organic TCP connections during the (short) time the attack executes may also thwart the attack; however, we integrate a mechanism for robustness against organic TCP connections. 
In~\autoref{sec:eval} we evaluate the attack under these conditions (devices connected through NATs and establishing organic connections while the attack executes).
Our technique cannot track clients that connect via forward proxies, which establish a new TCP connection to the tracking server (instead of a direct connection from the client). Particularly, it is ineffective against Tor clients.

\subsection{Attack Overview}
The attack exploits a core vulnerability in DHPS. DHPS assigns a source port to a destination (IP address and port -- $2^{48}$ combinations for IPv4) using a state maintained in one of the cells of its small perturbation table. This means that many destinations are generated using the same table cell, i.e.,  using a state that changes in a predictable way between accesses (DHPS increments the cell per each usage). The attack exploits this behavior for detecting collisions in the cell assignment hash function. Such collisions among loopback destinations are invariant to the network configuration of the device and can thus serve as a stable device ID.

To describe the attack, we first define two subsets of tuples that are of special interest for our tracking technique:
\begin{itemize}
    \item An {\em attacker 3-tuple} is a 3-tuple in which $\mathit{IP_{DST}}$ is the attacker's tracking server IP address, and $\mathit{IP_{SRC}}$ is the Internet-facing address used by the measured device.  
    \item A {\em loopback 3-tuple} is a 3-tuple in which $\mathit{IP_{DST}}$ is a fixed loopback address (e.g., 127.1.2.3), and $\mathit{IP_{SRC}}$ is a loopback-facing address used by the measured device (typically 127.0.0.1).  
\end{itemize}
The goal of the attack is to find collisions in $G_{K_2}$ for loopback 3-tuples. These collisions, described as pairs of loopback 3-tuples that hash to the same value, form the device ID. 

The attack consists of two phases. In the first phase (\autoref{alg:phase-1}), the attacker obtains $T$ attacker 3-tuples, each one corresponding to one cell of the perturbation table. The attacker does not know which 3-tuple maps to which cell, but that is immaterial to the attack. All the attacker needs is the existence of a 1-to-1 mapping between the perturbation table and the $T$ attacker 3-tuples. In the second phase, the attacker maps loopback 3-tuples into attacker 3-tuples (each loopback 3-tuple considered is mapped to the attacker 3-tuple that falls into the same perturbation table cell). This allows the attacker to detect collisions in $G_{K_2}$ for loopback 3-tuples.

\subsection{Phase 1}
\label{sec:phase1}
In this phase, the attacker obtains $T$ attacker 3-tuples so that each one corresponds to a unique cell in the perturbation table. This is done in iterations, as shown in \autoref{alg:phase-1}. Define $S'_0=\varnothing$. In iteration $i$, the attacker generates a set $S_i$ of {\em new} attacker destinations (in \extref{sec:phase1-analysis} we show that $|S_i|=T-1$ minimizes the number of phase 1 iterations).
The attacker then instructs the browser to send three bursts of TCP connection attempts (TCP SYN packets): the first burst to $S_i$, then the second burst to $S'_{i-1}$, and the third burst to $S_i$ again. 
An attacker 3-tuple in $S_i$ is determined to be {\em unique} if the difference between the two sampled source ports for that 3-tuple is 1, indicating that no other attacker 3-tuple in $S'_{i-1}$ or $S_i$ shares the same perturbation table cell. Define $V_i$ to be all such attacker 3-tuples in $S_i$, and define $S'_i=S'_{i-1} \cup V_i$. This is repeated until $|S'_i|=T$, i.e., all perturbation table cells are uniquely covered by the attacker 3-tuples in $S'_i$. 
\autoref{fig:phase1-illustration} illustrates a single iteration.

\begin{figure*}[ht!]
\centering
\includegraphics[trim=0cm 0.83cm 0cm 3.35cm, clip=true,width=\textwidth,height=\textheight,keepaspectratio]{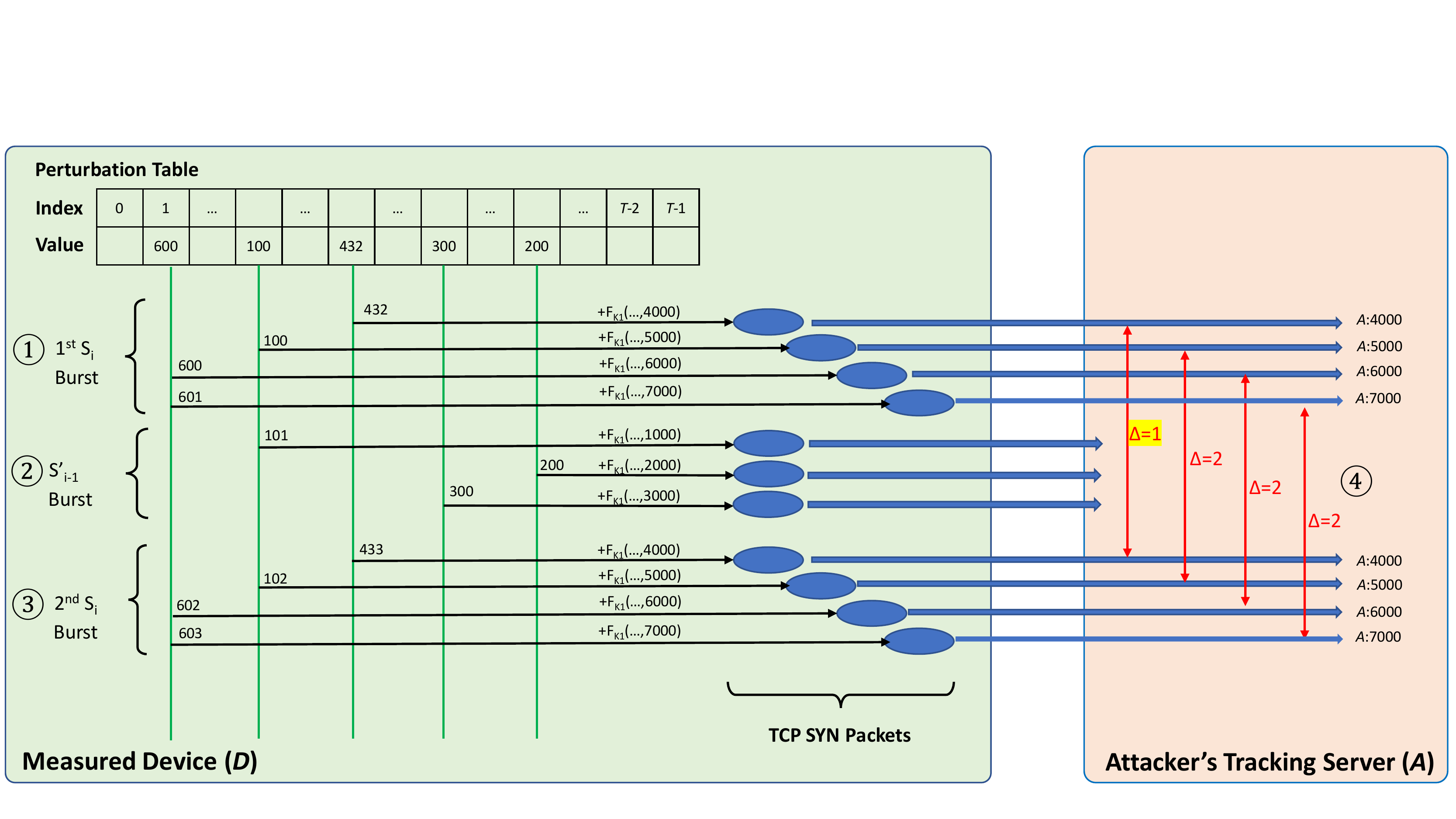}
\vspace{-5pt}
\caption{Phase 1 -- Single Iteration.
This example illustrates how the attacker adds 3-tuples which fall uniquely into cells. In Step \circled{1}, the device sends a first burst of TCP SYN packets for the new 3-tuple candidates, $S_i$ (4 in this example). In Step \circled{2}, the device sends the burst of TCP SYN packets for the set $S`_{i-1}$ of unique-cell attacker tuples (3 are shown in the illustration). In Step \circled{3}, the device sends a second burst of TCP SYN packets for the new 3-tuple candidates. In Step \circled{4}, the tracking server detects that only for the attacker 3-tuple which has destination port 4000, the source port was advanced by 1 (yellow background), which indicates that this 3-tuple has a unique cell. The attacker's 3-tuple with destination port 5000 had its source port advanced by 2 because it shares a cell with destination port 1000 in $S'_{i-1}$, and the 3-tuples with destination ports 6000 and 7000 share a cell; hence their source ports were advanced by 2.}
\label{fig:phase1-illustration}
\end{figure*}

In \autoref{sec:traffic-management}, we also show how phase 1 by itself can be used to measure the rate of outbound TCP connections. 

\begin{algorithm}[t]
\caption{Finding Attacker 3-Tuple per Cell (Phase 1)}
\label{alg:phase-1}
\begin{algorithmic}[1]
\Procedure{SendBurst}{$X$}
\ForAll{$x \in X$} 
\State $\Call{AttemptConnectTCP}{x}$
\EndFor
\EndProcedure

\Procedure{GetSourcePorts}{$U$}
\State \Call{SendBurst}{$U$}
\State $R \gets \Call{ReceiveAttackerTupleToPortMap}$ 
\State {\color{gray}\Comment $R = \{(\mathit{IP_{SRC}},\mathit{IP_{DST}},\mathit{PORT_{DST}}) \mapsto \mathit{PORT_{SRC}}\}$}
\State {\color{gray}\Comment (obtained from the tracking server)}
\State \Return $R$
\EndProcedure

\Procedure{Phase1}{}
\State $S' \gets \varnothing$
\While{$|S'|<T$}
\State $S_i \gets \Call{GetNewExternalDestinations}$ 
\State {\color{gray}\Comment $\forall_{j<i}(S_i \cap S_j = \varnothing)$}
\State $P \gets \Call{GetSourcePorts}{S_i}$ {\color{gray}\Comment \nth{1} burst}
\State $\Call{SendBurst}{S'}$ {\color{gray}\Comment \nth{2} burst}
\State $P' \gets \Call{GetSourcePorts}{S_i}$ {\color{gray}\Comment \nth{3} burst}
\State $S' \gets S' \cup \{x|P'(x)-P(x)=1\}$ 
\State {\color{gray}\Comment $V_i=\{x|P'(x)-P(x)=1\}$}
\EndWhile
\State \Return $S'$ 
\EndProcedure
\end{algorithmic}
\end{algorithm}

\subsection{Phase 2}
In the second phase, the attacker goes over a list $L$ of loopback 3-tuples. For each loopback 3-tuple, the attacker finds which attacker 3-tuple (one of the $T$ attacker 3-tuples found in phase 1) belongs to the same perturbation table cell. This mapping allows the attacker to find collisions in $G_{K_2}$ outputs among loopback 3-tuples, which (together with the number of iterations $l$) form the device ID. 

\begin{figure*}[ht!]
\centering
\includegraphics[trim=0cm 0.83cm 0cm 1.4cm, clip=true,width=\textwidth,height=\textheight,keepaspectratio]{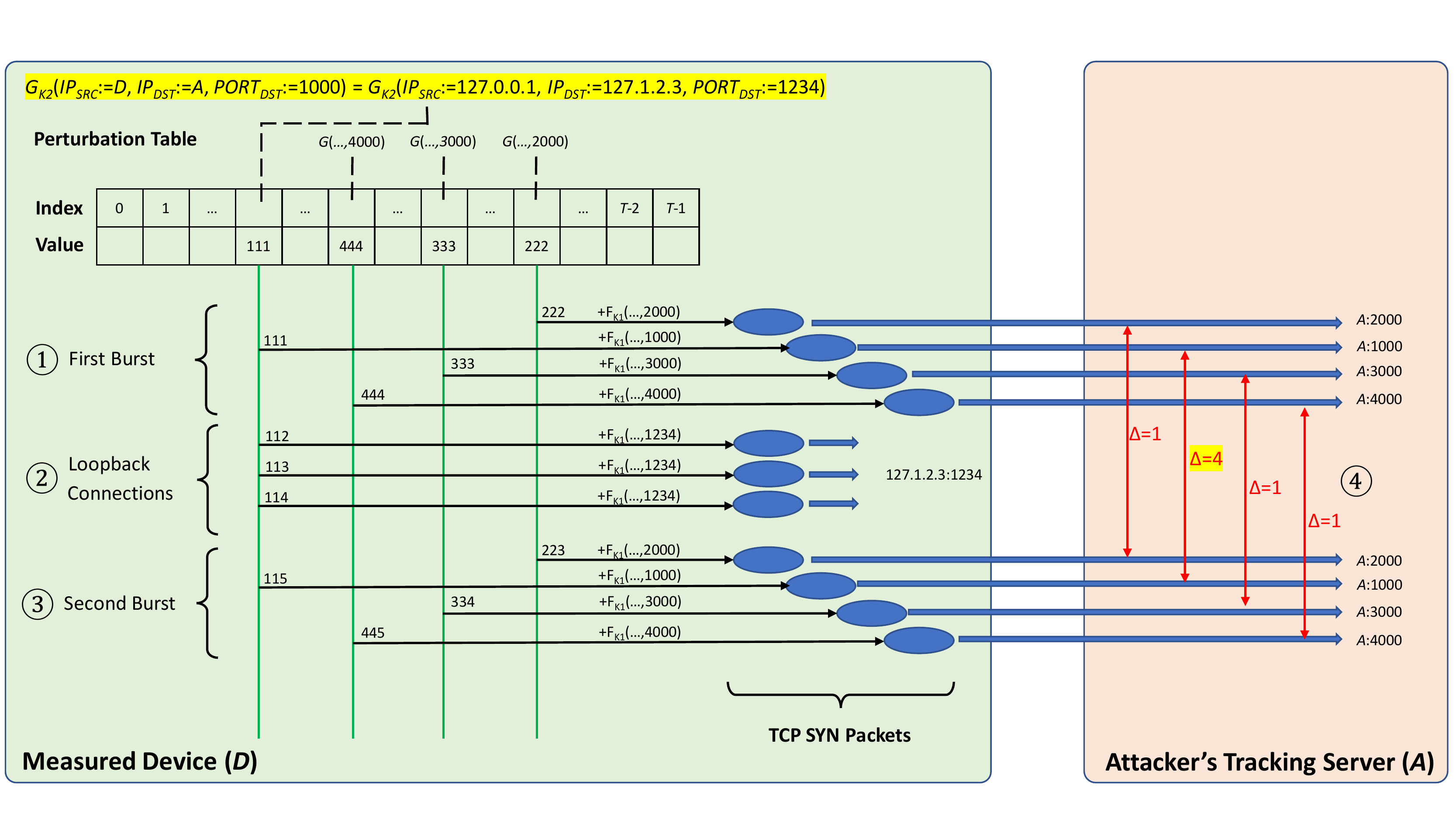}
\caption{Phase 2 -- Single Iteration.
This example illustrates how the attacker discovers that the the cell of the loopback 3-tuple $(\mathit{IP_{SRC}}=127.0.0.1, \mathit{IP_{DST}}=127.1.2.3, \mathit{PORT_{DST}}=1234)$ is identical to the cell of the attacker 3-tuple $(\mathit{IP_{SRC}}=A, \mathit{IP_{DST}}=D,\mathit{PORT_{DST}}=1000)$. In Step \circled{1}, the device sends a first burst of TCP SYN packets for all $T$ unique-cell attacker tuples (only 4 are shown in the illustration). In Step \circled{2}, the device sends several (3 in this example) TCP SYN packets to the loopback destination 127.1.2.3:1234. In Step \circled{3}, the device sends a second burst of TCP SYN packets for all $T$ unique-cell attacker 3-tuples. In Step \circled{4}, the tracking server detects that only for the attacker 3-tuple which has destination port 1000, the source port was advanced by at least 4 (right hand side, yellow background), which indicates that this 3-tuple shares the same counter (cell) with the tested loopback 3-tuple (yellow background equation at the upper left corner).}
\label{fig:phase2-illustration}
\end{figure*}

\begin{figure*}[t]
\centering
\includegraphics[trim=1cm 7.8cm 0.5cm 2.6cm, clip=true, width=\textwidth,height=\textheight,keepaspectratio]{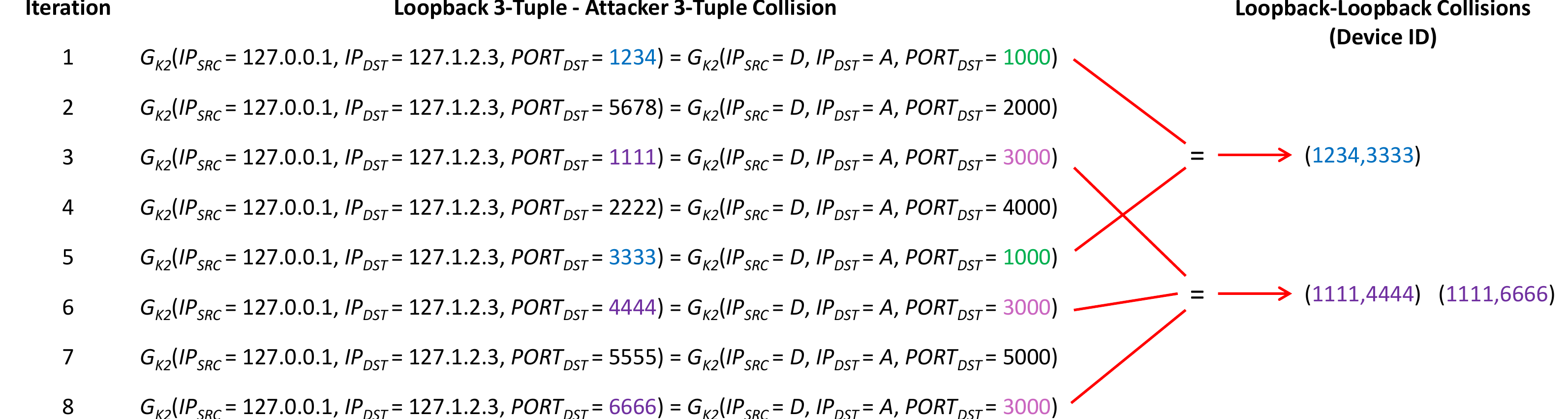}
\vspace{-8pt}
\caption{Phase 2 -- Calculating a Device ID from Multiple Iterations}
\label{fig:phase2-device-id}
\end{figure*}

In this phase (\autoref{alg:phase-2}), the attacker repeatedly runs iterations, until enough $G_{K_2}$ collisions are collected. In each iteration, the attacker maps a new loopback 3-tuple  $L_i$ to an attacker 3-tuple $w$, which hashes into the same cell of the perturbation table as the loopback 3-tuple. This is done by ``sandwiching'' a few loopback 3-tuple $L_i$ packets between bursts to all $T$ attacker 3-tuples obtained in phase 1, and observing which attacker 3-tuple $w$ has a port increment $>1$ (see \autoref{fig:phase2-illustration}). The attacker then collects new collisions with the first loopback 3-tuple in the cell ($B_w$) and counts the total number of collisions in $n$. This is illustrated in \autoref{fig:phase2-device-id}.

The attacker collects and counts ``independent'' colliding pairs. By this term we mean that if there are exactly $k$ loopback 3-tuples $x_1,\ldots,x_k$ that fall into the same cell, the attacker only uses $k-1$ pairs e.g. $(x_1,x_2),\ldots,(x_1,x_k)$, out of the possible $\binom{k}{2}$ pairs. 

\begin{algorithm}[t]
\caption{Finding a Device ID (Phase 2)}
\label{alg:phase-2}
\begin{algorithmic}[1]
\Procedure{Phase2}{}
\State $C \gets \varnothing$ \: ; \: $n \gets 0$ \: ; \: $i \gets 0$
\Repeat
\State $i \gets i+1$
\State $P \gets \Call{GetSourcePorts}{S'}$ {\color{gray}\Comment \nth{1} burst}
\State $\Call{AttemptConnectTCP}{\{L_i\}}$ 
\State $P' \gets \Call{GetSourcePorts}{S'}$ {\color{gray}\Comment \nth{2} burst}
\State $w \gets \:$\rotatebox[origin=c]{180}{$\iota$}$ x(P'_x-P_x>1)$ {\color{gray}\Comment $|\{x|P'_x-P_x>1\}|=1$}
\If{$\Call{Defined}{B_w}$} {\color{gray}\Comment A collision was found}
\State $C \gets C \cup \{(L_i,B_w)\}$ 
\State {\color{gray}\Comment Add the (single) independent pair to $C$}
\State $n \gets n+1$
\Else
\State $B_w \gets L_i$
\EndIf
\Until{$n \geq n^*_i$} 
\State {\color{gray}\Comment This is equiv. to $P^i_D(n) \leq p^*$ (\extref{sec:phase2-analysis})}
\State $l \gets i$
\State \Return $(C,l)$ 
\EndProcedure
\end{algorithmic}
\end{algorithm}

 Note that our attack makes no assumptions on the choice of \autoref{alg:alg-4}'s parameters (the hash functions $F$, $G$), beyond assuming that $G$ is reasonably uniform, which is guaranteed since RFC 6056 \cite[Section 3.3.4]{rfc6056} mandates that ``G() should be a cryptographic hash function''.
 
\subsubsection{Terminating with an Accurate ID}
\label{subsec:terminating}
We want phase 2 to terminate as soon as ``enough'' collisions are observed. By this, we mean that the probability of another device (i.e., a device with a random $K_2$) to produce the same set of collisions  is below a given threshold. Thus, we define our target function $P^l_D(n)$ to be the probability of a random device getting the same ID as the device $D$ at hand, which is assumed to terminate after $l$ iterations with exactly $n$ independent pairs. (We show below that $P^l_D(n)$ does not depend on the ``structure'' of the independent collisions, only on their number, i.e., it is well defined.) Note that $P^l_D(n)$ is defined for $l \geq 1$ and $\max(0,l-T) \leq n \leq l-1$. We can define $P_D=0$ elsewhere. 
We will then require $P^l_D(n) \leq p^*$, where $p^*$ is a threshold acceptance probability. As explained in \extref{sec:phase2-analysis}, the choice of $p^*$ depends on the expected device population size $N$, e.g. $p^*=\nicefrac{1}{\binom{N}{2}}$ guarantees there will be up to one ID collision on average in the entire population. 
We provide $p^*$ values for example population sizes in \extref{sec:phase2-analysis}.

To calculate $P^l_D(n)$, we begin by calculating the probability for a random device to have exactly the same set of collisions as device $D$, after $l$ iterations. Lemma \ref{lem:pldn_correctness} shows that this probability is in fact $P^l_D(n)$. Analyzing this probability is a ``balls in buckets'' problem. The buckets are the $T$ perturbation table cells (that map 1-to-1 to the $T$ attacker 3-tuples from phase 1), and the balls (in iteration $l$ of phase 2) are the loopback 3-tuples that are mapped to the perturbation table cells. However, when looking at two devices, $D$ and $D'$, the attacker has no way to map the buckets between the devices. The only information the attacker can consider is collisions among the balls, i.e., which balls (loopback 3-tuples) fall into the same bucket in both devices. We call this the ``structure'' of collisions. 
More formally, we define the structure of collisions in $D$ (after $l$ iterations) as follows: in a device $D$, we have $l-n$ occupied (non-empty) buckets $B_i$, and we have a set $C_i$ of loopback 3-tuples in $B_i$ such that $|C_i|=m_i$, and of course $\sum_i m_i=l$. The collision structure in $D$ then is the set of sets $\{C_1,\ldots,C_{l-n}\}$. 

For two devices $D'$ and $D$ to have the same structure, we look at the {\em first} loopback 3-tuple in each bucket -- denote $f_i$ as the first loopback 3-tuple in bucket $B_i$. When building the bucket for $D'$, we start with all-empty buckets. The first loopback 3-tuple $f'_1$ in $B'_1$ can pick any bucket, thus the probability to succeed in matching the structure of $D$ is $\frac{T}{T}=1$. The first loopback 3-tuple $f'_2$ in $B'_2$ has probability $\frac{T-1}{T}$ to match $D$'s structure, since it must not hit the first bucket $B'_1$, and so forth. So the combined probability of all first loopback 3-tuples in their buckets to match $D$'s structure is $\prod_{i=0}^{l-n-1} (1-\frac{i}{T})$. The remaining loopback 3-tuples must each go to its bucket in order to match, so each one has probability $\frac{1}{T}$. Therefore, 
$$P^l_D(n)=\frac{\prod_{i=0}^{l-n-1} (1-\frac{i}{T})}{T^n}$$
Interestingly, this probability does not depend on $(m_1,m_2,\ldots)$ of the structure -- it only depends on the total number of independent collisions, $n$. This means that $P^l_D(n)$ is well defined. 

We now show that $P^l_D(n)$ that we just calculated describes the probability for a random device to have the same device ID as D. This is proved by this lemma:
\begin{lemma}\label{lem:pldn_correctness}
If device D' has the same collision structure as in the device ID of device D after $l$ iterations, where $l$ is taken from the signature of $D$, then D' and D must have an identical device ID.
\end{lemma}
\begin{proof}
For device $D'$ to have the same device ID as device $D$, we also need the number of iterations to match -- i.e., we need to show that $l'=l$.  Recall that the phase 2 algorithm stops as soon as it reaches an iteration $l$ that fulfills the condition $P^l_D(n) \leq p^*$. Also, the order of loopback 3-tuples the algorithm tests is deterministic. Therefore, the algorithm for $D'$ will stop at exactly the same number of iterations as $D$.
\end{proof}

\subsection{Further Improvements}\label{subsec:device-tracking-further-improvements}

\paragraph{Phase 1 and 2 dwell time optimization.} 
The phase 1 algorithm, as depicted in \autoref{alg:phase-1}, is carried out entirely on the client-side. However, the client cannot obtain the TCP source ports of the TCP connections it attempts to establish with the tracking server (JavaScript code running in the browser has no way of accessing TCP connection information). This information is only available at the tracking server. Thus, the client and the server need to engage in a ``ping-pong'' of information exchange. This exchange is depicted in procedure $\textsc{GetSourcePorts}$ of \autoref{alg:phase-1}, where the client attempts to establish multiple TCP connections with the server, and the server returns a map from destinations to their TCP source ports. This ping-pong does not need to follow the exact form of \autoref{alg:phase-1}. In particular, there is no need to collect the TCP source ports of the first burst before sending the second burst. The two bursts can be sent one after another without pausing, as long as the client-side (OS) sends the bursts in the order the client prescribes. Similarly, the third burst can be sent by the client immediately after the second burst, since there is no need to synchronize with the server after the second burst is emitted.
Note that in \autoref{alg:phase-1}, proceeding to the next iteration requires the client to know which destination addresses to add to $S'$ (the set of unique attacker 3-tuples). 
Thus, moving {\em across iterations} requires the client and server to synchronize.

The same argument can be applied to phase 2 (\autoref{alg:phase-2}). In essence, this algorithm can run almost asynchronously between the client and the server. The client only ensures that the bursts are distinguishable at the server-side (see below).
We further optimize by moving the termination logic to the server and having the client run through the iterations over $L_i$ and only carry out TCP connection attempts. Thus, for each new iteration, the server collects the source ports and updates the attacker 3-tuple cell ($B_w$), the device ID ($C$), and the number of independent pairs ($n$). It also checks the termination condition and signals the client to stop when it is met.  
In this scheme, the client does not need to wait for the server's response after each burst. Hence, the client is free to make a significant optimization, entirely getting rid of the first burst in each iteration, since the server can use the previous ``second'' burst just as accurately. The phase 2 algorithm then becomes \autoref{alg:phase-2-optimized}.

\begin{algorithm}[t]
\caption{Finding a Device ID (Phase 2, Optimized)}
\label{alg:phase-2-optimized}
\begin{algorithmic}[1]
\Procedure{Phase2-Server}{}
\State $C \gets \varnothing$ \:; \: $n \gets 0$ \:; \: $i \gets 0$
\State $P\gets \Call{CollectSourcePorts}{S'}$
\Repeat
\State $i \gets i+1$
\State $P' \gets \Call{CollectSourcePorts}{S'}$ 
\State $w \gets \:$\rotatebox[origin=c]{180}{$\iota$}$ x(P'_x-P_x>1)$ {\color{gray}\Comment $|\{x|P'_x-P_x>1\}|=1$}
\If{$\Call{Defined}{B_w}$} {\color{gray}\Comment A collision was found}
\State {\color{gray}\Comment Add the (single) independent pair to $C$}
\State $C \gets C \cup \{(L_i,B_w)\}$ 
\State $n \gets n+1$
\Else
\State $B_w \gets L_i$
\EndIf
\State $P \gets P'$
\Until{$n \geq n^*_i$} 
\State {\color{gray}\Comment This is equiv. to $P^i_D(n) \leq p^*$ (\extref{sec:phase2-analysis})}
\State $\Call{SignalClientStop}$
\State $l \gets i$
\State \Return $(C,l)$ 
\EndProcedure

\Procedure{Phase2-Client}{}
\State $\Call{SendBurst}{S'}$
\For{$i=1$ to $l_{max}$}
\State $\Call{AttemptConnectTCP}{L_i}$
\State $\Call{SendBurst}{S'}$
\EndFor
\EndProcedure
\end{algorithmic}
\end{algorithm}

\paragraph{Burst separation.}
The optimization just described mandates the server to separate the traffic back into bursts since packets might be re-ordered by the network. Now that the client and server are not synchronized, this task is a bit more challenging. 
The client controls the timing on its end, i.e., a burst is followed by loopback traffic, followed by another burst, etc., precisely in this order. On the server-side, we can separate bursts using the source port itself. When we look at all the traffic to a certain destination (port), we can deduce the order of sending by sorting the packets by their TCP source port, which is monotonously increasing with the sending time (by the properties of the DHPS algorithm). So the packet with the lowest TCP source port (for a specific server port destination) belongs to the first burst, the second packet -- to the second burst, etc. 
That being said, source ports may wraparound if $\mathit{offset+{table}_{index}}$ exceeds $\mathit{num\_ephemeral}$. To account for that, we begin by normal sorting and then find the wraparound point by looking for a large enough difference that has been created as a result of the wraparound.

\paragraph{Minimizing number of bursts.}
We modify the phase 2 algorithm to group several loopback 3-tuples together, sandwiched between attacker-directed bursts. Let $\alpha$ be the number of loopback 3-tuples in a group. (In our implementation, $\alpha=4$.) Our modified algorithm reduces the number of phase 2 iterations by a factor of $\alpha$. Let $L_0,L_1,\dots,L_{\alpha-1}$ be the loopback 3-tuples in a single group. The idea is to vary the number of connections made for each loopback 3-tuple in a group, and have the tracking server differentiate between the loopback 3-tuples in the group by the magnitude of the difference between its two sampled measurements. Our algorithm makes $\beta \times 2^{i}$ connections to $L_i$ ($\beta=50$ in our implementation).

Denote by $\Delta_w$ the difference between the two source port measurements collected by the server for attacker 3-tuple $w$. With $\alpha=1$, it is the same as described until now: $L_0$ shares the same table cell as an attacker 3-tuple $w$ for which $\Delta_w=\beta+1$. With $\alpha=2$, there are two possibilities: either $L_0$ and $L_1$ each map to different attacker 3-tuples $w_0$ and $w_1$, or both collide with the same attacker 3-tuple $w$ (i.e. $w=w_0=w_1$). In the former case, the server will detect that $\Delta_{w_0}=\beta+1$ and $\Delta_{w_1}=2\beta+1$. The difference for $w_1$ is larger since the number of connections attempted for $L_1$ is twice as much as $L_0$. In the latter case, where $w_0=w_1=w$, the server will detect that for a single $w$ it has  $\Delta_w=1\beta+2\beta+1=3\beta+1$. Thus, it concludes that the two loopback 3-tuples in the group must share the same table cell. This argument can be extended to higher values of $\alpha$, see below.

\paragraph{Robustness against organic TCP connections.}
Organic TCP connections from the device spread uniformly across the $T$ cells of the perturbation table and are thus unlikely to significantly affect a single cell in the short time our attack runs. Yet, such connections noise the server's measurements
and we robustify our technique against such noise.

The phase 1 algorithm already handles noise: say $w\in S_i$ is a unique attacker 3-tuple and that some organic TCP connections were intertwined between the two source port measurements for $w$. If those TCP connections share the same table cell as $w$, then $\Delta_w>1$. In this case, the algorithm determines that $w$ is not unique. It does not compromise correctness: the algorithm continues iterating until covering all table cells. 

In phase 2, we cannot rely on an algorithm that searches for precise differences so we add safety margins. Specifically, we segment the \emph{differences space} to $2^\alpha$ disjoint segments $I_0,I_1,\dots,I_{2^\alpha-1}$, where $I_k=[k\beta+1,(k+1)\beta]$. Our algorithm then maps each attacker 3-tuple to a segment. The idea is that since the noise is typically small enough, it will cause the difference to be slightly above the ``noiseless'' expected value $k\beta+1$, but still below $(k\beta+1)+\beta$, i.e., $<(k+1)\beta+1$. In other words, the difference will still belong to the segment $[k\beta+1,(k+1)\beta]$. Thus, given a difference in segment $[k\beta+1,(k+1)\beta]$, we determine that it is a result of $k\beta$ loopback connections, and from the binary representation of $k$ we can reconstruct which $L_i$'s were mapped to this $w$. These are precisely the $L_i$'s in which $2^i$ appears as an addendum in the deconstruction of $k$ into a sum of powers of two.
To summarize, our phase 2 algorithm maps a loopback 3-tuple $L_i$ to attacker 3-tuple $w$ if and only if the $i$-th bit in the binary expansion of $k_w$ (the segment number of $w$) is one.

\subsection{Performance Analysis}
\subsubsection{Number of Iterations}\label{subsec:performance-analysis-iterations}
We analyze the run-time (in terms of iterations) of phase 1 and of phase 2 in \extref{app:analysis}.
For example, for $T=256$ (the Linux case), phase 1 (\autoref{alg:phase-1}) needs 13.8 iterations on average to conclude. For $T=256$ and a population of $N=10^6$ devices, in order for the average ID collision count to be lower than one for the entire population, phase 2 (\autoref{alg:phase-2} or \autoref{alg:phase-2-optimized}) needs 49.5 iterations on average to conclude. For a population of $N=10^9$, phase 2 takes on average 60.4 iterations to conclude, and for $N=10^{12}$, 69.9 iterations on average.

\subsubsection{Dwell Time}
In terms of dwell time (how long the browser needs to remain on the page for the process to complete), in phase 1 (\autoref{alg:phase-1}), each burst in a single iteration can be sent without waiting to the server's response. However, at the end of each iteration, the browser needs to wait for the server's response (the server updates the client on which ports should be used in the next iteration). Therefore, for phase 1, the required dwell time is the time needed for the browser to emit the 3 bursts, plus RTT, times the number of iterations. For phase 2 (\autoref{alg:phase-2-optimized}), the client and the server are not synchronized per iteration. Therefore, the dwell time is the time to emit the 2 bursts, times the number of iterations.

\subsubsection{Device Bandwidth Use}
In terms of packet count (and network byte count), iteration $i$ of phase 1 emits $|S_i|+|S'_{i-1}|+|S_i|$ packets. Since $|S_i|=T-1$ and $0 \leq S'_{i-1} \leq T-1$, we have the packet count in each iteration between $2(T-1)$ and $3(T-1)$.  Each iteration of phase 2 (\autoref{alg:phase-2-optimized}) consists of $T$ packets (we do not count the loopback packets as they do not consume physical network resources). Each packet in a burst is a TCP SYN which has a minimal size (Linux TCP SYN is 60 bytes over IPv4 and 80 bytes over IPv6).

%% file: sec-implementation.tex
\section{Implementation}
Recent Linux kernels (versions 5.12-rc1 and above) use DHPS from RFC 6056  to generate TCP source ports. The implementation includes a few modifications to the RFC algorithm and parameter choice (\S\ref{sec:linux}) that require some adaptations of our  technique (\S\ref{sec:impl:adapt}). We conclude this section by describing our implementation (\S\ref{sec:implementation-details}).

\subsection{Linux's DHPS Variant}
\label{sec:linux}

\paragraph{Perturbation table.}
The Linux implementation has $T$={\tt TABLE\_LENGTH}=256.
The values of the perturbation table are incremented by 2 (instead of 1), but this is immaterial to the attack. To simplify presentation, we ignore this detail for the rest of the discussion. Linux's implementation of \textsc{CheckSuitablePort} verifies that the $port$ is not locally reserved, and that the 4-tuple it forms, ($\mathit{IP}_\mathit{SRC},port,\mathit{IP}_\mathit{DST},\mathit{PORT}_\mathit{DST}$) is not already in use in the current network namespace (container). In Linux, $|K_2|=128$.

\paragraph{Noise injection.}
A significant modification in the Linux implementation is that it randomly injects noise to the perturbation table when {\tt \_\_inet\_hash\_connect()} finds a suitable candidate in the first iteration of \autoref{alg:alg-4}. The relevant table cell is then incremented twice (instead of once) with probability~$\frac{1}{16}$.

\subsection{Adapted Phase 2 Algorithm}\label{sec:impl:adapt}
The phase 1 algorithm already handles noise caused by organic TCP connections, as described in \autoref{subsec:device-tracking-further-improvements}, so noise injected by Linux is handled identically.
The adapted phase 2 algorithm should map the correct attacker 3-tuple for any given loopback 3-tuple, despite any noise injected by the kernel.
We already describe in \autoref{subsec:device-tracking-further-improvements} how the phase 2 algorithm can handle some noise (up to $\beta$ additional table cell increments on top of the expected ones, per a table cell associated with an attacker 3-tuple $w$). Recall, we strive to use as high as possible an $\alpha$ value since it reduces the run-time of phase~2. 
Thus, we seek to make $\alpha$ as large as possible while keeping $\beta$ higher than the noise induced by the device's Linux kernel and sporadic TCP connections.

The noise is maximal when $w$ corresponds to a table cell which has a total of $(2^\alpha-2)\beta$ connections. (The hardest task is to distinguish segment $[(2^\alpha-2)\beta+1,(2^\alpha-1)\beta]$ from $[(2^\alpha-1)\beta+1,2^\alpha\beta]$.) The amount of noise for $(2^\alpha-2)\beta$ has a binomial distribution $\mathrm{Bin}((2^\alpha-2)\beta,\frac{1}{16})$. To support a population of $10^6$ devices, we want an error probability of $10^{-6}$ to correctly find all the $L_i$ values in a given device, which requires all $\frac{64}{\alpha}$ iterations of phase 2 to succeed. (Our implementation tests 64 loopbacks in phase 2, see \autoref{sec:dwell-time}.) Therefore, we require $Prob(\mathrm{Bin}((2^\alpha-2)\beta,\frac{1}{16})\geq \beta) \leq \frac{\alpha}{64}10^{-6}$. For $\alpha=4$, this yields $\beta\geq 1244$ which is impractical since this implies $1244\times(2^4-1)=18\,660$ connections to be made for each group of loopback 3-tuples. 

We now show how we can tweak the algorithm to support $\alpha=4$ with a low $\beta$ value. For this, we observe that there is exactly one $w$ where $\Delta_w \geq 2^{\alpha-1}\beta$; it is exactly the $w$ that shares the counter with $L_{\alpha-1}$. All other $L_i$'s together cannot contribute more than $(2^{\alpha-1}-1)\beta$ to any counter. Thus, if we put the $w$ whose $\Delta_w \geq 2^{\alpha-1}\beta$ aside, we are left with differences which are $\leq (2^{\alpha-1}-1)\beta$. This upper-bound forms the worst case, with the noise distribution $\mathrm{Bin}((2^{\alpha-1}-1)\beta,\frac{1}{16})$. Again we require $Prob(\mathrm{Bin}((2^{\alpha-1}-1)\beta,\frac{1}{16}) \geq \beta) \leq \frac{\alpha}{64}10^{-6}$. This time, for $\alpha=4$, we get $\beta\geq 50$. Using $\beta=50$ results in a manageable number of connections. 

To summarize, we put aside the $w$ whose $\Delta_w \geq 2^{\alpha-1}\beta$, we find $L_i$'s that belong to all other $w$'s and associate all the remaining $L_i$'s (which were not associated to any other $w$) to the $w$ whose $\Delta_w \geq 2^{\alpha-1}\beta$. This technique allows using $\alpha=4$ with $\beta=50$ to support 1 million devices. The value $\beta$ grows slowly with the number of supported devices, e.g., supporting 1B devices requires $\beta = 73$.

\subsection{Device Tracking Technique Prototype}\label{sec:implementation-details}
We implemented a proof-of-concept of our device tracking technique. 
The client (snippet) is designed to run on both Chrome and Firefox and is implemented in approx. 300 lines of JavaScript code. 
We also implemented a client in Python that helped us during development and testing. The server is implemented in approx. 1000 lines of Go code. We describe key aspects of our implementation below and provide additional, more minor, implementation details in \autoref{sec:minor-implementation-details}.

The client and server exchange information via HTTP. In our implementation, the server acts as a ``command-and-control'' host: the client runs in a loop, requests the next command from the server, and executes it. This shifts much of the complexity from the client to the server and makes it easier to update the implementation without changing a complicated JavaScript implementation for multiple browsers.

\paragraph{Chrome and Firefox clients.}
Our client implementation for Google Chrome emits bursts of TCP connections by utilizing WebRTC. The client passes the list of destinations as a configuration for the {\tt RTCPeerConnection} object and triggers the burst with the {\tt setLocalDescription()} API. The advantage of WebRTC, compared to standard {\tt XMLHttpRequest} or {\tt fetch()} APIs, is that it allows to create connections at a rapid pace, which helps decreasing the overall dwell time. For Firefox, WebRTC is not applicable, since its implementation invokes {\tt bind()} on the sockets before it calls {\tt connect()} (see \autoref{sec:firefox}). Instead, our Firefox implementation uses {\tt XMLHttpRequest}.

\paragraph{Tracking server.}
The server uses {\tt libpcap} to capture TCP SYN packets. It associates the TCP SYN packet to an active tracking client based on the source IP address and destination port. The server also stores the source port for later processing. 
For any attack-related incoming TCP SYN, our server replies with TCP RST (+ACK), except for the HTTP/HTTPS port on which the attacker's web server listens, of course. This way, upon receiving the RST, the client immediately discards the socket and does not send any further packets on it. This also has the advantage of keeping Linux's kernel connection table relatively vacant -- otherwise, we risk hitting the process file descriptor limit (for Chrome, the limit is 8192).

\paragraph{Handling retransmissions.}
When a TCP SYN is left unanswered, Linux retransmits the SYN packet. This has the potential to confuse our server by having more source ports measurements than expected. The Linux retransmission timeout is 1 second ({\tt TCP\_TIMEOUT\_INIT} in {\tt include/net/tcp.h}). Therefore, we might encounter this situation depending on the client's network and the RTT to the tracking server. To cope with retransmissions, the server deduplicates the TCP SYN packets it receives based on the combination of source IP, source port, and destination port fields.

\input{app-minor-impl}

%% file: app-minor-impl.tex
\subsection{Minor Implementation Details}\label{sec:minor-implementation-details}

\subsubsection{Supporting Firefox}
\label{sec:firefox}
\paragraph{WebRTC vs. \texttt{XmlHttpRequest}}
Linux uses the revised DHPS for assigning TCP source ports when {\tt connect()} is invoked with an unbound socket. This is the standard practice for establishing TCP connections from a client. However, it is also technically possible to establish a TCP connection from a client by first invoking {\tt bind()} on a socket with a zero local port (instructing the kernel to pick a source port for the socket), and then applying {\tt connect()} to it. When {\tt bind()} assigns a source port number to the socket, the kernel has no information regarding the destination, and therefore it cannot use DHPS. Instead, Linux uses Algorithm 1 of RFC 6056 in this case. 
While not intended to be used by TCP clients, Firefox does in fact use {\tt bind()} for its WebRTC connections. Therefore, in Firefox we resort to using \texttt{XmlHttpRequest}, which emits HTTP/HTTPS requests.

\paragraph{Re-connect attempt at the HTTP level}
When Firefox's HTTP connection attempt fails or times out, Firefox retries the connection at the HTTP level, i.e. it tries to establish a new TCP connection (as opposed to a kernel retry which uses the original connection parameters). These additional TCP connection attempts with their newly-generated TCP source ports cannot be easily distinguished from the expected TCP connection attempts.
The phase 1 algorithm might be affected from re-connections when an attacker 3-tuple receives $>2$ source port measurements. In such a case, the algorithm cannot analyze the attacker 3-tuple, and should discard it (in the present round), so that it is would not be flagged as ``unique''.
The phase 2 algorithm was modified too: the client contacts the server after \emph{each} group of loopback 3-tuples are tested (instead of after all groups). Further, for each burst of TCP connections the client makes, it also \emph{waits} for each connection to complete (in contrast to the Chrome implementation). This prevents concurrent loopback connections and re-connections, which may corrupt source port measurements.  Computing $\Delta_w$ is a bit more complicated, since as a result of re-connections, we might get $>2$ source port measurements (more than expected). Consequently, we compute the difference between consecutive source ports measured, and set $\Delta_w$ to the maximum difference. 
We expect all consecutive differences to be 1 (up to some noise), except a single consecutive difference, which will correspond to the loopback increment we seek. This is because we do not allow re-connection attempts to overlap with the loopback connection attempts.
For example, if we get 4 source port measurements $s_0,s_1,s_2,s_3$ then we set $\Delta_w=\mathsf{max}(s_1-s_0,s_2-s_1,s_3-s_2)$. 
The algorithm continues normally from this point onward.

\subsubsection{Private Network Access}
The Private Network Access draft standard \cite{private-network-access}, which is implemented in Chrome, does not block our attack. If the snippet page is served over a secure context (HTTPS), then for XHR requests to loopback addresses, the browser will first attempt to send pre-flight HTTP OPTIONS requests to these destinations. These TCP connection attempts (SYN packets) to the loopback addresses suffice for our attack. Moreover, the Private network Access standard only applies to HTTP protocol traffic, and therefore WebRTC traffic to loopback addresses is not covered by it.

\subsubsection{Scalability}
In principle, the technique can scale very well. A technical limitation we need to consider is that we can only associate up to 65455 ports per $(\mathit{IP_{SRC}},\mathit{IP_{DST}})$ address pair (the port range is 1-65535, with some 80 ports blocked by Chrome). This is not an issue for clients coming from different IP addresses. However, it may be a limitation when multiple clients behind one IP (NAT). We can address it by using more server IP addresses and load balance clients. 

To maximize the gain from each server IP address, it should carefully manage the assignment and release of ports from a per-client-IP pool.
In phase 1, each client consumes at most $2(T-1)$ destination ports during each iteration. In each iteration of phase 1, a batch of $T-1$ new destinations (ports) can be obtained from the server in real time (assigned from the pool). When the iteration is complete, the server determines which destinations are added to the client's list, and which are released back to the pool. Phase 1's worst case for ongoing consumption is $T-1$ ports (the client's expanding list).
In phase 2, each client consumes exactly $T$ destinations (the full client list, from phase 1). This list of destinations is allocated to the client for several seconds (the duration of phase 2).
Thus, the peak consumption for each client is $2(T-1)$, so the algorithm can sustain $\frac{65455}{2(T-1)}$ simultaneous clients (with the same IP address) per server IP address. For $T=256$, the server can sustain 128 clients (behind NAT) per server IP. 

\subsubsection{Handling packet drops}
Packets may get dropped in a network due to congestion or routing changes. We adapt the tracking server to withstand moderate packet loss. First, we detect that there are too few packets for a specific attacker 3-tuple. Then, we find the burst for which the packet is missing by timing analysis: the largest time difference between captured packets is probably where a packet is missing. For the burst where the packet is missing, we may still find the attacker 3-tuple for loopback address $L_i$ by examining if there is a valid gap in TCP source ports of another attacker 3-tuple. If not, we rerun the test for this loopback address.

%% file: sec-evaluation.tex
\section{Evaluation}\label{sec:eval}

We use our proof-of-concept implementation to evaluate the device tracking technique against Linux-based devices. Our experiments answer the following questions:
{
\begin{enumerate}
  \setlength{\itemsep}{4pt}
  \setlength{\parskip}{0pt}
  \setlength{\parsep}{0pt}
    \item Do we get a consistent device ID across browsers, tabs, and browser privacy modes?
    \item Do we get a consistent device ID across networks, in both IPv4 and IPv6, across containers, and across VPNs?
    \item Do we get a consistent device ID when the user browses other sites during the attack?
    \item What is the dwell time required by our attack?
    \item Is the attack applicable to Android devices?
    \item How does the attack scale in terms of CPU and RAM? Is it suitable for large-scale tracking?
\end{enumerate}
}

\paragraph{Setup.}
We deploy our tracking server in two Amazon EC2 regions: {\tt eu-south-1} (IPv4 and IPv6) and {\tt us-west-2} (IPv4 only). Each server is a {\tt t3.small} instance, with 2 vCPU cores, 2GB of RAM and 5Gbps network link.

We tested three Ubuntu 20.04 Linux client devices: (i) HP Spectre x360 laptop (Intel Core i7-7500U CPU with 16GB of RAM) with kernel v5.13.19; (ii) ASUS UX393EA laptop (Intel Core i7-1165G7 CPU with 16GB RAM)  with kernel v5.15.11; and (iii) Intel NUC7CJYH mini-PC (Intel Celeron J4005 CPU, 8GB RAM) with kernel v5.15.8.

\subsection{Browsers}
We demonstrate that we get a consistent device ID with our client snippet on Google Chrome 
(v96.0.4664.110) and Mozilla Firefox (v96.0) browsers (the latest versions at the time of writing). Since Chrome dominates the browser market, our optimizations are geared towards it.

We tested Chrome with our two tracking servers, both on IPv4 and IPv6. We verified that we get a consistent device ID across multiple tabs and browser modes, i.e., regular mode and incognito mode (one of the goals of this mode is to bolster privacy by thwarting online trackers). 
For Firefox, we verified that the modified tracking technique as depicted in \autoref{sec:firefox} works over the Internet and that the ID is identical to the one obtained via Chrome (cross-browser consistency).

\subsection{Networks, NATs, VPNs and Containers}
Our attack targets the client device, which operates in a variety of environments. It might access the Internet via a VPN, run the browser in a container or use a network behind a NAT. This section evaluates whether, and to what extent, these environments affect our attack.

\paragraph{Networks and NATs.} We tested our attack on multiple networks, on both landline and cellular ISPs, with IPv4 and IPv6. \autoref{tab:tested-networks} summarizes our results. Our technique yielded a consistent ID for all tested networks which support IPv6. With IPv4, we found that some networks rewrite the TCP source port value (probably due to in-path port-rewriting NATs). Since our attack relies on observing the device-generated TCP source ports, it failed to obtain an ID on such networks. For IPv4 networks that do not rewrite TCP source ports, and for a given device, we got a consistent device ID, identical to the ID obtained from IPv6 networks for the same device (cross-network consistency, including cross IPv4/IPv6 consistency).

\begin{table}[t]
\centering
\begin{threeparttable}
\resizebox{1.0\columnwidth}{!}{
\begin{tabular}{l|l|l|l}
\toprule
Network                             & \begin{tabular}{@{}l@{}}Port \\ Rewriting?\end{tabular}  & Throttling?   & IPv4 / IPv6                  \\ \midrule
EduRoam
& No                & No            & \checkmark / NA                  \\ 
University Guest
& Yes               & No            & \xmark / NA               \\ 
Landline ISP 1 
& No                & No            & \checkmark / \checkmark                       \\ 
Landline ISP 2 
& No                & No            & \checkmark / NA                  \\ 
Landline ISP 3 
& No                & No            & \checkmark / NA                   \\ 
Landline ISP 4 
& No                & IPv4 only     & \checkmark\tnote{*} / \checkmark                 \\ 
Cable ISP 1 
& Yes               & No            & \xmark / NA                   \\ 
Cellular ISP 1 
& IPv4 only         & IPv4 only            & \xmark / \checkmark                         \\ 
Cellular ISP 2 
& IPv4 only         & IPv4 only            & \xmark / \checkmark                          \\ 
Cellular ISP 3 
& No                & Yes           & \checkmark\tnote{*} / NA                 \\ 
\bottomrule
\end{tabular}
}
\begin{tablenotes}\footnotesize
\item[*] With slowed-down TCP SYN bursts.
\end{tablenotes}
\end{threeparttable}
 \caption{
 Tested networks
}
\label{tab:tested-networks}
\end{table}

NATs are generally deployed on customer premise equipment (CPE) or at the ISP; the latter is often referred to as carrier-grade NAT (CGN). Importantly for our attack, many NAT implementations preserve the clients' TCP port selection.
Mandalari et al. \cite[Table I]{natwatcher} showed that 78\% of the tested NATs preserve TCP source ports. Their study covers over 280 ISPs that used 120 NAT vendors. Richter et al. \cite{multi-perspective-cgn} found that 92\% of the CPE NATs they have identified in non-cellular networks preserve TCP source ports. For CGN deployments, Richter et al. found that on cellular networks, port preservation is less common: about 28\% of the cellular networks with CGNs exhibit such a behavior. Among the non-cellular ISPs that use CGNs, 42\% preserve ports. These measurements are in-line with our tested networks in \autoref{tab:tested-networks}.

To further assess the applicability of our attack, we deployed servers on four cloud providers (Azure, AWS, Google Cloud, and Digital Ocean) across 25 different regions and tested whether the network rewrites TCP source ports under IPv4.  (IPv6 connections are less likely be NATed, as illustrated by our measurements in \autoref{tab:tested-networks}, since IPv6 addresses are abundant.) 
Our experiment shows that all tested regions and cloud providers do \emph{not} rewrite TCP source ports. We list the cloud providers and regions we tested in \autoref{sec:cloud-providers-experiment}.

Another issue we faced, mainly with cellular IPv4 ISPs, is traffic throttling (see \autoref{tab:tested-networks}). Such networks limit the packet rate of our TCP SYN bursts by dropping some packets. This may be due to traffic shaping or security reasons (SYN flood prevention). To address this problem, we spread our TCP SYN bursts over a longer time, thus increasing the overall dwell time of our technique.

\paragraph{VPNs.} We tested our technique with a client device that is connected to an IPv4 VPN. We examined two popular VPN providers: TunnelBear VPN and ExpressVPN. In both cases, we set up a system-wide VPN configuration using Ubuntu's Network Manager. We note that the vast majority of VPNs (in particular, the two VPNs we tested) do not support IPv6 \cite{IPv6-NAT}.

We tested two locations with TunnelBear VPN, Germany and Finland, and found that the TCP source ports are preserved. However, TunnelBear's exit nodes throttle the outbound VPN traffic, so we expect the attack to work with slowed-down bursts. 
For ExpressVPN, our attacks succeeded on 7 out of 10 exit nodes tested in North America and Western Europe. The failed attempts were due to TCP port rewriting. The device ID we obtained through VPNs was identical to the one obtained when the device was using a regular network, and the dwell time we experienced with ExpressVPN was comparable to a regular network with the same RTT. We conclude that, in many cases, VPN exit nodes do not rewrite TCP source ports, which allows our technique to work. This demonstrates that VPNs do not inherently protect against our technique. 

\paragraph{Linux containers.} We deployed two docker containers on the same host and ran our Python client implementation on each. Both runs produced the same device ID, identical to the host device ID, as expected (cross-container consistency). We conducted the experiment with {\tt containerd} version 1.4.12, {\tt runc} version 1.0.2 and {\tt docker-init} version 0.19.0.

\subsection{Active Devices}
In this experiment, we demonstrate that a consistent device ID is obtained when the client simultaneously visits other websites during the attack. To this end, we opened multiple tabs on Chrome during the attack and visited sites that are listed under Alexa's Top 10 Sites. On each test, we arbitrarily chose 3 to 4 websites from the list (this includes ``resource-heavy'' sites such as Yahoo, YouTube, and QQ). In all of our tests, we verified that we get a consistent device ID, concluding that our technique successfully withstands organic TCP connections generated by the victim device during the attack.

\subsection{Dwell Time}\label{sec:dwell-time}
The dwell time is the execution time of our attack, the sum of phase 1 and 2 completion times. Phase 1 completion time is affected by the number of iterations it takes to collect $T$ unique attacker 3-tuples. The measured average is 15 iterations, which is expected to be slightly higher than the theoretical average of 13.8 iterations we computed in \extref{sec:phase1-analysis} due to the Linux injected noise (see \autoref{sec:linux}). The network round trip time (RTT) to the tracking server also affects phase 1's completion time since, at the end of each iteration, the client contacts the server to determine which of the attacker 3-tuples tested in this iteration are unique. 
Phase 2's completion time is mainly affected by the number of loopback groups tested. In our implementation, we tested a fixed number of 64 loopback groups. The RTT has significantly less impact on phase 2 since we do not wait for server responses.

We measured the dwell time for our Chrome client when using $\alpha=4$ and $\beta=50$ (see \autoref{subsec:device-tracking-further-improvements}) against our two tracking servers. 
With an average RTT of 50ms, we measured an average dwell time of 7.4 seconds. With an average RTT of 275ms, we measured an average dwell time of 13.1 seconds. Overall, our results for Chrome show a dwell time of 5-15 seconds (10 seconds on average).

For Firefox, the dwell time is on the order of several minutes, even under lab conditions, because of in-browser throttling. For this reason, our implementation uses more balanced values for $\alpha,\beta$: $\alpha=2$ and $\beta=10$. Lowering the $\beta$ value for Firefox has an advantage since the number of connections to the loopback interface (on phase 2) decreases, hence making the attack run faster under throttling. We have not attempted optimizing our Firefox implementation further.

\subsection{Android}\label{sec:android}
At the time of writing, there is no Android device that uses DHPS (see the discussion on \autoref{sec:rfc6056}). To verify that the attack is applicable to Android, we manually introduced the DHPS code into the 5.4 kernel of a Samsung Galaxy S21 device. (Linux's DHPS implementation is conveniently located in one file, making this change self-contained.)

During our experiments with the Samsung device, we observed a {\tt netfilter} rule that limits the rate of incoming new TCP SYN packets to an average of 50 per second. This rule limits our attack in phase 2 when many connections to the loopback interface are attempted (an outgoing TCP SYN packet in the loopback interface eventually becomes an incoming TCP SYN packet, which is subject to this rule). To work around this restriction, we modified the $\alpha,\beta$ parameters of the attack to $\alpha=2$ and $\beta=10$. This results in $10+2\times10=30$ loopback connections being attempted for each loopback group, which is below the limit imposed by this rule. 
With this configuration, our lab experiments show that the attack on Chrome yields a consistent device ID when the device switches networks, with a dwell time of 18-21 seconds.

We presume that the offending rule (causing us to adjust the attack) is not general to Android but rather, it is Samsung-specific since we 
did not find evidence of it in the Android Open Source Project (AOSP) code.

\subsection{Resource Consumption}
\paragraph{Server side CPU utilization.}
The tracking server needs very little CPU resources. In both phases, the server-side calculation simply involves going over the iteration data and finding the port pairs in which the difference is above some threshold.

\paragraph{Server side RAM consumption.} 
While computing an ID for a device, the tracking server needs to keep (i) a list of allocated destination ports; and (ii) the most recently observed source port per each destination port. In the worst case, the list can be up to $2T$ pairs of ports (with $T=256$, the maximum list length is 512). Each port number is 16 bits, so each port pair takes 4 bytes, and overall the server needs up to $8T=2$KB per actively tested device. Suppose the device dwell time is $W$ seconds, and the rate of devices per second incoming funnel is $R$, then at a given moment, there will be $R\cdot W$ tested devices, which requires $8T\cdot R \cdot W$ bytes.
For Linux-based devices using Chrome, we have $T=256, W=10s$, so the tracking server RAM consumption is $R \cdot 20$KB per device. For example, if $R=10^6$ devices/sec, the server needs 19.07GB RAM. Keep in mind that $R$ does not represent the total number of devices the server needs to support, it is only the rate at which the server is required to measure devices (which is much lower than the total number of supported devices).

\paragraph{Client side CPU utilization.} 
In our experiments, CPU utilization by our tracking logic was negligible.

\paragraph{Client side RAM consumption.}
Most of the client-side RAM consumption imposed by our tracking snippet is due to attempting to create TCP connections. In the attack, the tracking server responds with TCP RST, prompting the victim’s browser and kernel to release this memory quickly. In an experiment, we measured an overhead of at most 30MB in client memory (tested on Chrome), which is sufficiently low to allow the attack even on low-end devices.

%% file: sec-recommendations.tex
\section{Countermeasures}\label{sec:countermeasures}
The root cause of our attack is the ability to detect hash collisions in $G_{K_2}$ via a shared perturbation table cells. 
The ideal solution would be to create a private perturbation table instance per network namespace and interface. By doing so, we prevent the attacker's ability to detect ``stable'' hash collisions (on the loopback interface). 
The problem with this solution is that its memory consumption could be high when many containers are spawned, or many interfaces are present. 

To mitigate the attack when the perturbation table is shared across interfaces, DHPS must ensure that either hash collisions are much less frequent or that detecting collisions is much more difficult. In line with the above, we propose below modifications to DHPS and summarize how they were applied to the Linux kernel in a recent patch to mitigate our attack. We also analyze in \autoref{app:algorithm5} an alternative algorithm proposed by RFC 6056, ``Random-Increments Port Selection Algorithm'' (Algorithm 5, \cite[Section 3.3.5]{rfc6056}), showing that the trade-off it offers cannot simultaneously meet the functionality and security goals from ~\cite[Section 3.1]{rfc6056} (see ~\autoref{sec:rfc6056}). 

\ignore{
To mitigate the attack when the perturbation table is shared across interfaces, DHPS must ensure that either hash collisions are much less frequent or that detecting collisions is much more difficult. We note that RFC 6056 proposes an alternative algorithm ``Random-Increments Port Selection Algorithm'' (Algorithm 5, \cite[Section 3.3.5]{rfc6056}). We analyze this algorithm in \autoref{app:algorithm5} and find that the trade-off it offers cannot simultaneously meet the functionality and security goals. 

In line with the above, we propose below modifications to DHPS 
\begin{revision}and summarize how they were applied to the Linux kernel in a recent patch to mitigate our attack\end{revision}:
}

\paragraph{Increase the table size $T$.} This makes hash collisions much less frequent. For example, instead of $T=256$ we can use $T=256K=262,144$. This consumes 0.5MB of RAM (assuming each table entry is 16 bits). The attack now takes $\times1024$ time due to the need to cover all $T$ table cells. (In the patch issued for Linux, $T$ was increased to $64K$.)

\paragraph{Periodic re-keying.} Changing the secret key in DHPS results in a different table index being accessed and a different port offset from which candidate enumeration will begin. After re-keying, any table collision information previously obtained by an attacker becomes useless. The trade-off is that a source port that was previously chosen 
for the same 3-tuple could be chosen again. This might prevent a TCP connection from being established (see \autoref{sec:rfc6056}). To reduce the chance for connectivity issues, re-keying should not be performed too frequently. (In the patch issued for Linux, re-keying is performed every 10 seconds, 
balancing functionality and security.)

\paragraph{Introduce more noise.} The Linux kernel team also increased the noise in perturbation table cell increments to make detecting collisions more difficult. Now, each increment is an integer chosen uniformly at random between 1 and 8.

\subsection{Network Security Measures}
Besides mitigating our attack at its core (the DHPS algorithm), network security appliances such as firewalls, IDS, or IPS could thwart it to some degree. Since our attack requires a non-negligible number of connection attempts to the same set of attacker destinations, a network security appliance could detect this condition and limit future connection attempts to those destinations. This would require the attacker to slow the rate of connection attempts, presumably to the point where the attack becomes impractical. Crucially, however, by doing so, the security appliance might flag legitimate traffic as malicious: for example, when multiple users behind NAT enter a resource-heavy website in a short time span. 

A better mitigation strategy would be to use \emph{on-host} logic, which can be part of a {\em host} IPS (HIPS) solution, a personal firewall, or in-browser security logic. The on-host logic can detect the internal loopback traffic that is generated as a result of our attack and rate limit TCP connection attempts to {\em closed} ports on the loopback interface. We do not expect standard applications to exhibit similar behavior, making this countermeasure effective and with a low false-positive rate.

%% file: sec-conclusions.tex
\section{Conclusion}
This paper illustrates a flaw in the DHPS algorithm for selecting TCP source ports -- an algorithm which was proposed by RFC 6056 and recently adopted by Linux. We exploit this algorithm to track Linux devices across networks, browsers, browser privacy modes, containers, and VPNs. 
The key observation of this paper is that the attacker can detect {\em collisions} in the output of a keyed hash function used in DHPS via sampling, in an attacker-prescribed manner, TCP source ports generated from the global kernel state induced by DHPS. By indirectly observing such collisions for {\em loopback} 3-tuples inputs, the attacker calculates an invariant device ID (agnostic to the network the device is connected to) since it only relies on the secret hashing key which is generated at system startup. 
Interestingly, this result does not rely on the choice of the hash function, because the hash function's output space is small enough for collisions to happen naturally even in a very moderate number of TCP connections.
We implement the attack for Linux devices and Chrome and Firefox browsers, and demonstrate it across the Internet in practical, real-life scenarios. We propose changes to the DHPS algorithm that mitigate the root cause of our technique. Lastly, we worked with the Linux kernel team to integrate countermeasures into the Linux kernel, which resulted in a recent security patch.

%% file: sec-vendor-status.tex
\section{Vendor Status}\label{sec:vendor-status}
We reported our findings to the Linux kernel security team on February \nth{1}, 2022.
In response, the Linux team developed a security patch which was incorporated in versions 5.17.9 (and above) and 5.15.41 (the LTS versions that include DHPS). The Linux kernel issue is tracked as CVE-2022-32296.

%% file: app-analysis.tex
\section{Analysis of the Runtime of Our Attack Against DHPS}\label{app:analysis}
Our basic work unit (or time consumption unit) is a burst of up to $T$ packets. This is because the algorithm works by sending bursts, waiting for the burst to be received, sampling, sending another burst, etc.
The time it takes to send the burst is typically lower than the RTT, which is why the burst size has little impact on the overall runtime.

\subsection{Calculating the Distribution of \texorpdfstring{$\mu_r$}{mu[r]}}

We begin with calculating the distribution of $\mu_r(l)$ -- given $T$ bins, we define the variable $\mu_r(l)$ as the number of bins that have exactly $r$ balls in them, after $l$ balls are thrown at random into the bins. In later sections we will need the distribution for $\mu_0$ (also known as the ``occupancy problem'') and $\mu_1$, but our calculation technique is generic and can be used to accurately and efficiently calculate any $\mu_r$. 
For the special case $r=0$ (known as the ``occupancy problem''),  there is an explicit description (following the notation in \cite{doi:10.1080/00031305.2019.1699445}): $Prob(\mu_0(l)=k)=\mathit{Occ}(T-k|l,T)=\frac{1}{T^l}\binom{T}{k}\sum_{i=0}^{T-k}\binom{T-k}{i}(-1)^i((T-k)-i)^l$
. However, we are not familiar with any explicit description of this distribution for $r>0$, though there is literature that describes moments of the distribution, and some limits \cite{nla.cat-vn867984}. Here we describe an algorithm that calculates the distribution accurately and efficiently for any $T,l,r$. Our technique for calculating $\mu_r$ provides $\mu_{r'}$ for $r'<r$ as a by-product.

The idea is to calculate the probability $\mu_l^*(i_0,i_1,\ldots,i_r)$ of $T$ bins to have exactly $i_s$ bins with $s$ balls in them (where $\sum_{s=0}^r i_s \leq T$) after $l$ balls are thrown. We start with $\mu_0^*(T,0,\ldots,0)=1$ and $\mu_0^*=0$ elsewhere. We now look at a transition from $(i_0,\ldots,i_r)$ after $l-1$ balls. For $s \leq r$, the $l$-th ball is thrown into an a bin with $s$ balls with probability $\frac{i_s}{T}$, thus contributes $\mu_{l-1}^*(i_0,\ldots,i_s,i_{s+1},\ldots,i_r)\frac{i_s}{T}$ to $\mu_{l}^*(i_0,\ldots,i_s-1,i_{s+1}+1,\ldots,i_r)$. Conversely, the $l$-th ball is thrown into a bin with more than $r$ balls with probability $\frac{T-\sum_{s=0}^r i_{s}}{T}$, thus contributes $\mu_{l-1}^*(i_0,\ldots,i_r)\frac{T-\sum_{s=0}^r i_{s}}{T}$ to $\mu_{l}^*(i_0,\ldots,i_r)$.
The vector of probabilities $\mu^*_l$ is then calculated by going over all $(i_0,\ldots,i_r)$ events for ball $l-1$ and summing their contributions to the various $(i'_0,\ldots,i'_r)$ events for ball $l$. Specifically, we can write:

$$\mu^*_l(i_0,\ldots,i_r)=$$
$$(\sum_{j=0}^{r-1}\mu^*_{l-1}(i_0,\ldots,i_j+1,i_{j+1}-1,\ldots,i_r)\frac{i_j+1}{T})+$$
$$\mu^*_{l-1}(i_0,\ldots,i_{r-1},i_r+1)\frac{i_r+1}{T}+$$
$$\mu^*_{l-1}(i_0,\ldots,i_r)\frac{T-\sum_{s=0}^r i_{s}}{T}$$

(where we define $\mu^*_{l-1}(i_0,\ldots,i_r)=0$ if $\sum_{s=0}^r i_s>T$, or $i_s>T$ or $i_s<0$ for any $s$).
Finally, to calculate the distribution of bins with $r$ balls after $l$ balls are thrown, we simply sum all the events which have exactly $k$ bins with $r$ balls:
$$Prob(\mu_r(l)=k)=\sum_{i_0,\ldots,i_{r-1},i_r=k}\mu^*_l(i_0,\ldots,i_r)$$ 
By the same notion, we also have, for $r'<r$:
$$Prob(\mu_{r'}(l)=k)=\sum_{i_0,\ldots,i_{r'-1},i_{r'}=k,i_{r'+1},\ldots,i_r}\mu^*_l(i_0,\ldots,i_r)$$ 

The number of $(i_0,\ldots,i_r)$ combinations is exactly the number of combinations of  $r+2$ non-negative numbers that add up to $T$, which is $\binom{T+r+1}{r+1}$. Thus the memory complexity is $\binom{T+r+1}{r+1}$ (for fixed $r$, this is $O(T^r)$). Since each event at iteration $l$ involves a calculation based on (up to) $r+2$ events at iteration $l-1$, the runtime complexity is $l(r+2)\binom{T+r+1}{r+1}$ (for fixed $r$ and $l$, this is $O(T^r)$). Specifically, for $r=1$, we have memory complexity of $\frac{(T+2)(T+1)}{2}$ and runtime complexity of $1\frac{1}{2}(T+2)(T+1)l$. Likewise, for $r=0$ (the occupancy problem), we have runtime complexity of $2(T+1)l$ and memory complexity of $T+1$.

\subsection{Analysis of Phase 1}
\label{sec:phase1-analysis}
\subsubsection{Optimal size of new candidate batch}
Note: for simplicity, herein we ignore the (rare) events wherein out-of-order transmission cause an attacker 3-tuple $y$ (which does not collide with any of the attacker 3-tuples in $S'_{i-1}$) to be accepted into $S'_i$ even though other attacker 3-tuples in $S_i$ collide to its cell. In other words, we assume that for $y$ to be accepted into $S'_i$, $y$'s cell must not collide with the cells of both $S'_{i-1}$ {\em and} $S_i \setminus \{y\}$. 

Suppose we already have $n$ cells (out of $T$) already covered from previous iterations, i.e. $|S'_{i-1}|=n$. We want to test $k=|S_i|$ new external destinations in the current iteration, and maximize the expected additional cells we cover, namely new cells which contain exactly one new external destination. For this, we define the variable $G(n,k)=|S'_i \setminus S'_{i-1}|$. The question is, therefore, which $k$ maximizes $E(G(n,k))$ (the {\em optimal} $k$), and whether it depends on $n$.

We first calculate the probability of a single cell to contribute to $G(n,k)$, i.e. for a cell to become filled with exactly one external destination (more accurately -- attacker 3-tuple, but we use it here interchangeably). For this, the cell must be initially empty (w.r.t. external destinations), which has the probability $\frac{T-n}{T}$. Now, since we have $k$ new attacker destinations, The probability of a particular cell to contain a single {\em new} attacker destination is the binomial term, $\binom{k}{1}(\frac{1}{T})^1(1-\frac{1}{T})^{k-1}$. Therefore, the probability of an arbitrary cell to contribute to $G(n,k)$ is $\frac{T-n}{T}\binom{k}{1}(\frac{1}{T})^1(1-\frac{1}{T})^{k-1}$.
We define a variable $X_i$ which is 1 if the cell $i$ is not already covered, and has one new attacker destination, and 0 otherwise, we have $E(X_i)=\frac{T-n}{T}\binom{k}{1}(\frac{1}{T})^1(1-\frac{1}{T})^{k-1}$. We now ask what is the expectancy of $G(n,k)$ -- the number of cells containing a single external destination, out of all $T$ cells. Clearly $E(G(n,k))=E(\sum_i X_i)$. Since expectancy is additive regardless of dependencies, i.e. it is 
$E(G(n,k))=T\frac{T-n}{T}\binom{k}{1}(\frac{1}{T})^1(1-\frac{1}{T})^{k-1}=\frac{T-n}{T}k(1-\frac{1}{T})^{k-1}$.
Indeed, this can be optimized over $k$ without regard to $n$.
The optimal $k$ is obtained when $k(1-\frac{1}{T})^{k-1}$ is minimal, i.e. $\frac{d}{dk}k(1-\frac{1}{T})^{k-1}=0$. This can be easily solved: $k=-\frac{1}{\ln{(1-\frac{1}{T})}}$. More precisely, since $k$ should be integral, We need to choose between $k_L=\lfloor-\frac{1}{\ln{(1-\frac{1}{T})}}\rfloor$ and $k_H=\lceil -\frac{1}{\ln{(1-\frac{1}{T})}}\rceil$ -- the one that provides the higher $E(G(n,k))$. 

For high values of $T$, since $\ln{(1-\frac{1}{T})} \approx -\frac{1}{T}$ (and always $\ln{(1-\frac{1}{T})} < -\frac{1}{T}$) we have $k_L=T-1$ and $k_H=T$. In such a case, it is easy to see that $E(G(n,T))=E(G(n,T-1))$, therefore both $k$ values yield the same result. There is a slight advantage to using less packets, so we choose $k=T-1$. Therefore, our implementation of $\textsc{GetNewExternalDestinations}$ can be as depicted in \autoref{alg:destinations}, wherein the destinations are different ports (starting at port $p_I$) on the same single attacker IP address.

\begin{algorithm}[t]
\caption{Allocate Attacker Destinations}
\label{alg:destinations}
\begin{algorithmic}[1]
\State $i \gets 0$
\Procedure{GetNewExternalDestinations}{}
\State $i \gets i+1$
\State return $\{\mathit{{IP}_{SRV}}:(p_I+(i-1)(T-1)),\ldots,\mathit{IP_{SRV}}:(p_I+i(T-1)-1)\}$
\EndProcedure
\end{algorithmic}
\end{algorithm}

\subsubsection{Distribution Function for the Number of Iterations}
We define the probability distribution for the number of covered cells at the end of iteration $l$ ($l$ starts at 1),  ($|S'_l|$), to be $p_l$. We think of $p_l$ as a vector of $T$ entries, where $p_l[i]$, $0<i\leq T$ is the probability that exactly $i$ cells are covered. We define $P(k)=Prob(\mu_1(T-1)=k)$ as the probability for having exactly $k$ bins with a single ball when throwing $T-1$ balls into $T$ bins.
Define $S^*_l \subseteq S_l$ to be destinations in $S_l$ which do not fall into the same cell with any other destination in $S_l$ (i.e. they belong to a unique cell). If we know that $|S^*_l|=k$, and we know that $|S'_l|=j$, then the probability for $|S'_{l+1}|=i$ (where $i \geq j$) is the probability that exactly $i-j \leq k$ balls (out of $k$) from $S^*_l$ do not fall into $S'_l$. 

To calculate this, we first need to choose the subset of $i-j$ balls that do not fall into $S'_l$. We have $\binom{k}{i-j}$ ways to choose it. Then, the first ball in the subset needs to fall into the $T-j$ bins not covered by $S'_l$, which has the probability $\frac{T-j}{T}$. The second ball in the subset is known to fall into a different bin, hence has a probability of $\frac{T-j-1}{T-1}$ to fall into a bin not covered by $S'_l$, and so forth. The total probability of the $(i-j)$ -sized subset to fall outside $S'_l$ (i.e. in $\bar{S'_l}$, which has $|\bar{S'_l}|=T-j$) is therefore $\frac{(T-j)(T-j-1)\cdots(T-j-(i-j-1))}{T(T-1)\cdots(T-(i-j-1))}$. Now, we also need the complementary subset of $k-(i-j)$ balls to fall into $S'_l$. For this to happen, we clearly require the complementary subset size to be no bigger than $|S'_l|$, i.e. $k-(i-j) \leq j$ which incurs $k \leq i$. Under this assumption, the probability of the complementary subset to fall into $S'_l$ is  $\frac{j(j-1)\cdots(j-(k-(i-j)-1))}{(T-(i-j))\cdots(T-(i-j)-(k-(i-j)-1))}$, and note that the first ball in this subset starts with $T-(i-j)$ possible bins (left from the first subset). Thus, the overall probability  $p_{ijk}$ of the event is (after further algebraic simplification):
$$p_{ijk}=\binom{k}{i-j}\frac{((T-j)\cdots(T-i+1))\cdot (j\cdots(i-k+1))}{T\cdots(T-k+1)}$$
Therefore, for $l>1$, $p_{l+1}[i]=\sum_{j=1}^{\min({i},{T-1})}p_l[j]\sum_{k=i-j}^{i}P(k)p_{ijk}$, and  $p_1[i]=P(i)$. Note that we deliberately do not sum the terms at $j=T$, since we want these to ``disappear'' from the system after the iteration in which they are created. In this way, $\sum_{i=1}^T p_l[i]$ represents the probability of the process to survive $l$ iterations, and is strictly monotonously decreasing, with $\sum_{i=1}^T p_1[i] =1$ and $\sum_{i=1}^T p_\infty[i]=0$. Also note that $p'(l)$, the probability of the process to terminate at iteration $l$ is simply $p_l[T]$, i.e. if we define $\varphi=(0,\ldots,1)$ as the linear functional that extracts the last vector coordinate, then $p'(l)=\varphi\cdot p_l$.
 
We see that $p_{l+1}$ is a linear combination with fixed coefficients of $p_l$, therefore we can write $p_{l+1}=A \cdot p_l$, where $A \in M(T,T)$ and $A_{ij}=\sum_{k=i-j}^{i}P(k)p_{ijk}$ for $j\leq \min(i,T-1)$ and $A_{ij}=0$ elsewhere.

Thus, we have the distribution $p'(l)$ for termination at iteration $l$:
$$p'(l)=\varphi(A^{l-1}P)$$
From this we can calculate the distribution of $l$, and specifically $E(l)$. The distribution function for $l$ is plotted in Fig.~\ref{fig:dist-phase1}. For $T=256$, we calculated $E(l)=13.819116$. 
\begin{figure}[t]
\centering
\begin{center}
\begin{tikzpicture}
\begin{semilogyaxis}[
    xlabel={Iterations ($l$)},
    ylabel={Probability},
    xmin=5, xmax=37,
    ymin=1e-5, ymax=1,
    xtick={5,10,15,20,25,30,35},
    ytick={1e-5,1e-4,1e-3,1e-2,1e-1},
    legend pos=north east,
    ymajorgrids=true,
    grid style=dashed,
]

    \addplot table[col sep=comma] {distribution_T256_CLEAN.txt}; 
    
    \addplot table[col sep=comma] {out.sim_phase1_N=1000000_T=256_distribution_CLEAN.txt}; 
        
    \legend{Theoretic Results,Empiric Results}

\end{semilogyaxis}
\end{tikzpicture}
\end{center}
\caption{Distribution Function for Phase 1 Iterations (only values $\geq 10^{-5}$ are shown)}
\label{fig:dist-phase1} 
\end{figure}
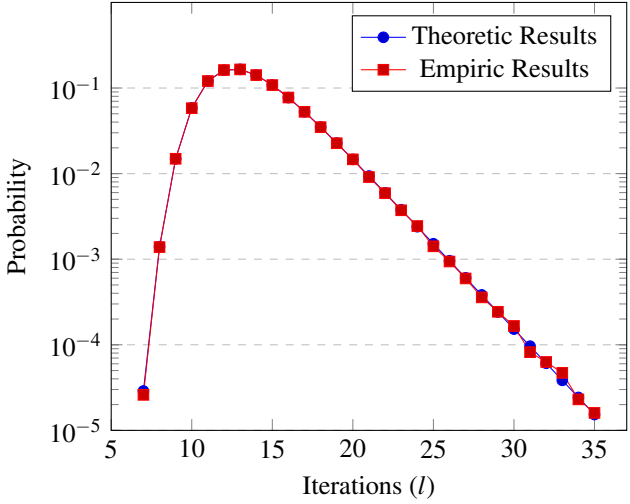

We also simulated Phase 1 for $T=256$ (with $M=1000000$ experiments). The average number of iterations needed was 13.81484, in line with the theoretic result. We also plotted the empiric distribution function in \autoref{fig:dist-phase1}. As can be seen, the empiric results fit the theoretic prediction very tightly.

\subsection{Analysis of Phase 2}
\label{sec:phase2-analysis}
Denote by $N$ the population size we want to support, in the sense that we allow an average of up to $c^*$ device ID collisions in a population of $N$ devices.
For simplicity, we assume that each device among the population of $N$ devices has a unique key ($K_2$). Of course, when $2^{|K_2|} \gg \binom{N}{2}$ this assumption is statistically valid, but the reader should keep in mind that the $K_2$ key space size is an upper bound on the device ID space, namely the device ID space cannot exceed $2^{|K_2|}$ even if we obtain the full list of collisions for the device at hand (note that the device ID space can reach almost $2^{|K_2|}$ with enough samples).

In general, we can add loopback destinations until we get ``enough'' independent pairs. We can define a threshold probability $p^*$ which is the maximal accepted probability for a random device to have the same device ID as the current device. We set $p^*=\frac{c^*}{\binom{N}{2}}$ -- this is an upper bound, since the number of independent collisions may be lower. For every new destination, we update $n$, the total number of pairs, and check the following condition (where $l$ is the number of loopback destinations):
$$P^l_D(n) \leq p^* $$
We stop at the first $l$ that satisfies the above condition. The device ID is then the $n$ independent pairs (and $l$) and the runtime is $l$.\\
Note: in order for the device ID to be well defined, we need to define a canonical selection of independent collisions. This can be done easily, for example, we choose only pairs where one of the destinations is the first one that was associated with the cell. Since the order of associating destinations to cells is deterministic (it is a subset of the order in which the destinations are enumerated), this canonical form is well defined.

Note that $P^l_D(n)$ is monotonously decreasing in $n$ and $l$, when $n \geq 0$ and $1 \leq l \leq T$. Also note that $P^l_D(n) \leq \frac{1}{T^n}$. Therefore, if $\frac{1}{T^n} \leq p^*$, it follows that $P^l_D(n) \leq p^*$. From this we can calculate an upper bound for $n$, which is $n_\mathit{UB}=\lceil-\log_T{p^*}\rceil$. It is important to stress that this is not a tight upper bound. Still, it shows that $n$ is typically very small, even for extremely small $p^*$. 

We want to calculate $l_{max}$ and $E(l)$. We assume $l \leq\lfloor T-\sqrt{T} \rfloor +1$ (we will see below how this upper bound is obtained).

For $l_{max}$, the worst case is $n=0$ (this is only so when $l\leq T$). This results in $P^l_D(0)=\prod_{i=0}^{l-1} (1-\frac{i}{T})$. This expression is monotonously decreasing for $1 \leq l \leq T$. We can therefore find the minimal $l_{max}$ such that $P^{l_{max}}_D(0) \leq p^*$ using binary search. Note that if $T$ is ``too small'' (i.e. if $\frac{T!}{T^T} > p^*$) then a solution may not be found, see \autoref{sec:low-T} for a discussion on this topic.

For $E(l)$, we first define $n^*_l$, which is the minimal $n$ (given $l$) that satisfies $P^l_D(n) \leq p^*$. We can find it using binary search with $0 \leq n \leq \min(n_\mathit{UB},l-1)$. It is easy to see that $n^*_l$ is monotonously decreasing in $l$ (if $P^l_D(n^*_l) \leq p^*$ for $l$, then $P^l_D(n^*_l) \leq p^*$ for $l'>l$, therefore $n^*_{l'} \leq n^*_l$). The termination condition for step $l$ is then simply that $n \geq n^*_l$ ($n$ includes step $l$'s potential collision). Note that this allows a simpler implementation of the algorithm, since it requires simply keeping track of the number of independent collisions, and comparing them to $n^*_l$ (which can be computed once, offline, per $p^*$). This leads to the following observation: \textbf{the algorithm is completely determined by the vector $n^*_l$ (where $l_{min} \leq l \leq l_{max}$, and we can define $n^*_l=\infty$ for $l<l_{min}$)}.\footnote{Technically, the algorithm also depends on $l_{min}$ and $l_{max}$ but these can be easily deduced from the $n^*_l$ series: $l_{min}=\min\{l|n^*_l<\infty\}$ and $l_{max}=\min\{l|n^*_l=0\}$.} In Table~\ref{tab:ns} we provide $n^*_l$ for $N=1000000,c^*=1$.

\begin{table}[t]
\begin{center}
\begin{tabular}{l|l}
\toprule
                     & $n^*_l$ \\ \midrule
$6 \leq l \leq 52$    & 5        \\ 
$53 \leq l \leq 72$  & 4        \\ 
$73 \leq l \leq 87$  & 3        \\ 
$88 \leq l \leq 98$  & 2        \\ 
$99 \leq l \leq 108$ & 1        \\ 
$l = 109$           & 0        \\ \bottomrule
\end{tabular}
\end{center}
\caption{$n^*_l$ for $N=1000000,c^*=1$}
\label{tab:ns}
\end{table}

Consider now the probability of the process to terminate at step $l$. For this, the process must arrive at step $l-1$ but not terminate there. But since $n$ is monotonously increasing in $l$, and $n^*_l$ is monotonously decreasing in $l$, it follows that if the condition for step $l-1$ is {\em not} satisfied, i.e. $n_{l-1}<n^*_{l-1}$ then for every $1 \leq l' < l$, $n_{l'}<n^*_{l'}$. Therefore, failing the condition for step $l-1$ guarantees that the process also did not terminate anytime earlier.

Lemma: For $l \leq \lfloor T-\sqrt{T} \rfloor +1$ and $n+1 \leq l-2$, $P^{l-1}_D(n+1) \leq P^l_D(n)$.\\
Proof: $\frac{P^{l-1}_D(n+1)}{P^l_D(n)}=\frac{1}{T(1-\frac{l-n-2}{T})(1-\frac{l-n-1}{T})} \leq \frac{1}{T(1-\frac{\lfloor T-\sqrt{T}}{T} \rfloor)(1-\frac{\lfloor T-\sqrt{T} \rfloor}{T})} \leq 1$.\\
Corollary 1: Assuming $l_{max}\leq \lfloor T-\sqrt{T} \rfloor +1$, $l \leq l_{max}$ and that $n^*_{l-1}$ exists (finite), then $n^*_{l-1} \leq n^*_l+1$. This is because $P_D^{l-1}(n^*_l+1) \leq P_D^l(n^*_l) \leq p^*$. Note that since $n^*$ is monotonously decreasing, $n^*_{l-1}=n^*_l$ or $n^*_{l-1}=n^*_l+1$.\\
Corollary 2: $n^*_{l_{min}} = l_{min}-1$. Clearly $n^*_{l_{min}} \leq l_{min}-1$. We prove that $n^*_{l_{min}} \geq l_{min}-1$ by contradiction. Suppose $n^*_{l_{min}}<l_{min}-1$ i.e. $n^*_{l_{min}}-1 \leq l_{min}-3$, then $P^{l_{min}-1}_D(n^*_{l_{min}}+1) \leq P^{l_{min}}_D(n^*_{l_{min}})$, and since $P^{l_{min}}_D(n^*_{l_{min}}) \leq p^*$ we have $P^{l_{min}-1}_D(n^*_{l_{min}}+1) \leq p^*$. Thus we can find $n^*_{l_{min}-1} \leq n^*_{l_{min}}+1$, in contradiction to the minimality of $l_{min}$.

We will henceforth assume $l \leq \lfloor T-\sqrt{T} \rfloor +1$. Again, this does not hold for ``small'' $T$ values, i.e. when $P^{\lfloor T-\sqrt{T} \rfloor +1}_D(0)>p^*$.

We will also need to calculate the probability of $l$ loopback destinations to have exactly $n$ independent collisions. This is equivalent to asking, when throwing $l$ balls into $T$ bins, what is the probability for having exactly $m$ empty bins. The number of independent collisions is then $n=l-(T-m)$. We define, therefore:
$$p_l(n)=Prob(\mu_0(l)=T+n-l)$$

We now calculate $p'(l)$, the probability to terminate at step $l$, based on the relationship between $n^*_l$ and $n^*_{l-1}$. Specifically, it is easy to see that $p'(1)=0$ (also for $l=1$, we only have $n=0$ and thus $p_1(0)=1$ and $P^1_D(0)=1$). Since there are only three options for $n^*_{l-1}$: undefined (when $l=l_{min}$), $n^*_l$ or $n^*_l+1$, we can calculate $p'(l)$ as follows:
\begin{itemize}
    \item If $l=l_{min}$ then $n^*_l=l_{min}-1$. Since in the previous round $l'=l_{min}-1$ the maximum number of pairs is $l'-1=l_{min}-2$, then for the algorithm to terminate, we must have $n_{l'}=l_{min}-2$ pairs already, and then a new pair is must be added. Therefore:
    $$p'(l_{min})=p_{l_{min}-1}(l_{min}-2)\frac{(l_{min}-1)-(l_{min}-2)}{T}=$$
    $$\frac{p_{l_{min}-1}(l_{min}-2)}{T}$$
    \item If $l>l_{min}$ and $n^*_l=n^*_{l-1}$ then in order to terminate at step $l$ but not at step $l-1$, we need to have $n_{l-1}=n^*_{l-1}-1$ (note that $n^*_{l-1}>0$, or else we would have stopped at $l-1$, therefore $n_{l-1} \geq 0$, this is captured in the fact that $p_{l-1}(\cdot)=0$ for negative arguments). We also need a new collision so that $n_l=n^*_l$, whose probability is simply $\frac{(l-1)-(n^*_{l-1}-1)}{T}$. The probability in this branch is, therefore:
    $$p'(l)=p_{l-1}(n^*_{l-1}-1)\frac{(l-1)-(n^*_{l-1}-1)}{T}$$
    \item If $l>l_{min}$ and $n^*_l=n^*_{l-1}-1$, then in order to terminate at step $l$ but not at step $l-1$, we need to have either $n_{l-1}=n^*_{l-1}-1$ or $n_{l-1}=n^*_{l-1}-2$ to fail the termination condition for step $l-1$. If $n_{l-1}=n^*_{l-1}-1$ then it does not matter whether there is a collision or not at step $l$, because the condition for step $l$ is immediately met. 
    This event has probability $p_{l-1}(n^*_{l-1}-1)$. If $n_{l-1}=n^*_{l-1}-2$ (this is impossible if $n^*_{l-1}-2$ is negative, i.e. if $n^*_{l-1}=1$, but this is captured in the fact that $p_{l-1}(\cdot)=0$ for negative arguments), then we also need a new collision so that $n_l=n^*_l=n^*_{l-1}-1$. This event has probability $p_{l-1}(n^*_{l-1}-2)\frac{(l-1)-(n^*_{l-1}-2)}{T}$. The total combined probability is:
    $$p'(l)=p_{l-1}(n^*_{l-1}-1)+$$
    $$p_{l-1}(n^*_{l-1}-2)\frac{(l-1)-(n^*_{l-1}-2)}{T}$$
\end{itemize}

We can define $l_{min}$, which is the minimal $l$ for which there is $n \leq l-1$ such that $P_D^l(n) \leq p^*$. Since $P_D$ is monotonously decreasing in $n$, the $n$ that minimizes it is $n=l-1$, thus $l_{min}$ is minimal such that $P_D^l(l-1) = \frac{1}{T^{l-1}}\leq p^*$. Therefore $l_{min}=\lceil 1-\log_T{p^*} \rceil$. It is then guaranteed that $p'(l)=0$ for $l<l_{max}$, i.e. the algorithm will not terminate for $l<l_{min}$.

When $n \geq n^*_l$, the process terminates. Note that inequality can happen when $n^*_l=n^*_{l-1}-1$ and $n_{l-1}=n^*_{l-1}-1$. In such case, another collision may be added in step $l$, ending with $n_l=n^*_{l-1}=n^*_l+1$.

Note that for the above calculated $l_{max}$, $n^*_{l_{max}}=0$, so this guarantees that the algorithm will terminate at most after $l_{max}$ steps.

Based on the above analysis, we can calculate $E(l)=\sum_{l=2}^{l_{max}} p'(l)\cdot l$.

Of interest is also the average collision count $c$ across the entire population. In the above, we bounded $c$ from above: $c \leq c^*$, since $c^*$ represents the average collision count when each and every device contributes a probability up to $p^*$.
If we extend the definition of $p'$ and define $p'(l,n)$ to be probability to terminate at step $l$ with $n$ pairs, then it is easy to see that the probability $\bar{p}$ of two random devices to have the same device ID is:
$$\bar{p}=\sum_{l,n}p'(l,n)P^l_D(n)$$
Therefore:
$$c=\binom{N}{2}\bar{p}=\binom{N}{2}\sum_{l,n}p'(l,n)P^l_D(n)$$

To calculate $c$, we need $p'(l,n)$ and $P^l_D(n)$. We already calculated $P^l_D(n)$ in \autoref{subsec:terminating}. We now calculate $p'(l,n)$: 
\begin{itemize}
    \item If $l=l_{min}$ then only $n=l-1$ is possible (see above), and we have:
    $$p'(l_{min},l_{min}-1)=p'(l_{min})$$ 
    \item If $l>l_{min}$ and $n^*_l=n^*_{l-1}$ then only $n=n^*_l$ is possible, and we have:
    $$p'(l,n^*_l)=p'(l)$$ 
    \item If $l>l_{min}$ and $n^*_l=n^*_{l-1}-1$, then $n$ can be either $n^*_l$ or $n^*_l+1$. The case of $n=n^*_l+1$ can only happen when there were $n^*_l=n^*_{l-1}-1$ pairs in the $(l-1)$-th round, and another pair was added in the $l$-th round, i.e.: $$p'(l,n^*_l+1)=p_{l-1}(n^*_{l-1}-1)\frac{(l-1)-(n^*_{l-1}-1)}{T}$$
    The case of $n=n^*_l$ can arise from two separate events: it is possible to have $n^*_l=n^*_{l-1}-1$ pairs in the previous round, and add no new pairs, or to have $n^*_l-1=n^*_{l-1}-2$ pairs in the previous round, and add a new pair. Therefore:
    $$p'(l,n^*_l)=p_{l-1}(n^*_{l-1}-1)(1-\frac{(l-1)-(n^*_{l-1}-1)}{T})+$$
    $$p_{l-1}(n^*_{l-1}-2)\frac{(l-1)-(n^*_{l-1}-2)}{T}$$
    
\end{itemize}
    
\subsubsection{Simulation Results}
We simulated the algorithm for population sizes $N=10^i$ (where $i=2,\ldots,12$). For each population size, we ran 500 simulations of an entire $N$ devices population, i.e. we sampled $500N$ $l$ values, and 500 $c$ values. The results are summarized in Table~\ref{tab:phase2-analysis-table}. As can be seen, the simulation results closely match our analysis.

\begin{table*}[t]
\centering
\begin{tabular}{l|l|l|l|l|l|l|l}
\toprule
$N$      & $p^*$                  & $l_{min}$ & $l_{max}$ & $E(l)$ & $\nicefrac{c}{c^*}$ & Sim. $E(l)$ & Sim. $\nicefrac{c}{c^*}$ \\ \midrule
$10^{2}$ & $2.020\times 10^{-4}$  & 3         & 64        & 30.027 & 0.16872                   & 30.053      & 0.164                          \\ 
$10^{3}$ & $2.002\times 10^{-6}$  & 4         & 78        & 35.151 & 0.37544                   & 35.153      & 0.348                          \\ 
$10^{4}$ & $2.000\times 10^{-8}$  & 5         & 90        & 39.261 & 0.35504                   & 39.261      & 0.334                          \\ 
$10^{5}$ & $2.000\times 10^{-10}$ & 6         & 100       & 44.899 & 0.19315                   & 44.899      & 0.21                           \\ 
$10^{6}$ & $2.000\times 10^{-12}$ & 6         & 109       & 49.496 & 0.24641                   & (UNTESTED)  & (UNTESTED)                     \\ 
$10^{7}$ & $2.000\times 10^{-14}$ & 7         & 117       & 53.01  & 0.33046                   & (UNTESTED)  & (UNTESTED)                     \\ 
$10^{8}$ & $2.000\times 10^{-16}$ & 8         & 124       & 56.6   & 0.2851                    & (UNTESTED)  & (UNTESTED)                     \\ 
$10^{9}$ & $2.000\times 10^{-18}$ & 9         & 131       & 60.354 & 0.26165                   & (UNTESTED)  & (UNTESTED)                     \\ 
$10^{10}$ & $2.000\times 10^{-20}$ & 10 & 137 & 63.843 & 0.25072 & (UNTESTED) & (UNTESTED) \\ 
$10^{11}$ & $2.000\times 10^{-22}$ & 11 & 143 & 66.891 & 0.27247 & (UNTESTED) & (UNTESTED) \\ 
$10^{12}$ & $2.000\times 10^{-24}$ & 11 & 149 & 69.917 & 0.2659 & (UNTESTED) & (UNTESTED) \\ \bottomrule
\end{tabular}
\caption{Phase 2 Analysis vs. Simulated Results ($c^*=1$)}
\label{tab:phase2-analysis-table}
\end{table*}

\subsection{Notes About Low \texorpdfstring{$T$}{T} Values}
\label{sec:low-T}

The above analysis assumes that $T$ is ``large enough'', i.e. that we will always get enough information for some $l \leq T$. We now explicitly calculate the threshold $T$ for this condition to hold.
Let us denote $l^*=\lfloor T-\sqrt{T} \rfloor +1$. If $P^{l^*}_D(0)>p^*$ then we cannot assume the algorithm to terminate with $l^*\leq \lfloor T-\sqrt{T} \rfloor +1$. 
We have:
$$P^{l^*}_D(0) = \frac{T\cdot(T-1)\cdots(T-l^*+1)}{T^{l^*}}=$$
We  use Stirling's approximation to obtain:
$$\frac{T!}{T^T}\frac{(T-l^*)^{T-l^*}}{(T-l^*)!}\Big(\frac{T}{T-l^*}\Big)^{T-l^*} \approx$$ $$\sqrt{2\pi T}e^{-T}\frac{e^{T-l^*}}{\sqrt{2\pi (T-l^*)}}\Big(\frac{T}{T-l^*}\Big)^{T-l^*}$$
We can also approximate $T-l^* \approx \sqrt{T}$:
$$P^{l^*}_D(0) \approx \sqrt[4]{T}e^{-l^*}\sqrt{T}^{\sqrt{T}}=e^{-(T-\sqrt{T}-\frac{1}{2}\sqrt{T}\ln{T}-\frac{1}{4}\ln{T})}$$

Therefore, for the analysis to be correct, we require that
$T-(1+\frac{1}{2}\ln{T})\sqrt{T}-\frac{1}{4}\ln{T} \geq -\ln{p^*}$.

Since $p^*$ is typically set to e.g. $\frac{2}{N^2}$, it follows that we typically need $T-(1+\frac{1}{2}\ln{T})\sqrt{T}-\frac{1}{4}\ln{T} \geq 2\ln{N}-\ln{2}$. We use $N=10^{12}$ as an upper bound (over $100 \times$ the current world population); the requirement then translates into $T\geq 86$. 

This does not mean low $T$ systems cannot be attacked. Quite the contrary. We simply do not confine ourselves to $l\leq T-\sqrt{T}+1$. This may complicate the analysis in general, but we can still easily analyze the extreme case $T=2$. In this case, a device measured provides exactly $L-1$ bits of information when $L$ loopback addresses are used: each one of the $L-1$ loopback addresses after the first either collides with the first or not, contributing exactly one bit. Therefore, we have $P^L_D=2^{-(L-1)}$ and given $p^*$, we need $L=\lceil -\log_2{p^*} \rceil+1$ iterations to satisfy the requirement. For example, with $N=10^6$, we have $p^*=2.000\times 10^{-12}$ and $L=40$, which is in fact considerably better than $T=256$ (whose $E(l)=49.496$).

%% file: sec-other-use-cases.tex
\section{Another Use Case: Traffic Measurement}
\label{sec:traffic-management}
It is possible to count how many outbound TCP connections are established by a device in a time period using the above techniques. This is useful for remote traffic and load analysis, e.g. to compare the popularity of services. For example, it is possible to remotely measure the rate at which outbound TCP connections are opened by a forward HTTP proxy. This can be used to estimate how many concurrent clients the HTTP forward proxy serves. And quoting \cite{Linux-patch-RFC6056}: ``In the context of the web, [counting] how many TCP connections a user's computer is establishing over time [...] allows a website to count exactly how many subresources a third party website loaded. [...] Distinguishing between different users behind a VPN based on distinct source port ranges, [...] Tracking users over time across multiple networks, [...] Covert communication channels between different browsers/browser profiles running on the same computer, [... and] Tracking what applications are running on a computer based on the pattern of how fast source ports are getting incremented''.

This attack builds on the phase 1 technique (\autoref{sec:phase1}). To mount the attack, the attacker needs to have client access to the forward proxy. For simplicity we assume that the attacker is simply one of the proxy's clients. The attacker first runs the phase 1 logic with the client side being a standalone script/software (not inside a browser) that runs on a machine that has client access to the forward proxy, and establishes $T$ attacker 3-tuples that conform to the $T$ perturbation table counters. This needs to be done only once, and ideally when the target device is relatively idle.

Next, the attacker can poll the TCP source ports $p_i$ for each of these $T$ attacker 3-tuples at time $t$. The attacker then polls the TCP source ports $p'_i$ at time $t'>t$. Denote by $\rho$ the total number of ephemeral ports in the system:  $\rho=\mathit{max\_ephemeral}-\mathit{min\_ephemeral}+1$, and suppose no counter advanced more than $\rho-1$ steps (keep in mind that the attacker's first poll also increments each counter by 1), then the total number of TCP connections established by the device between $t$ and $t'$ is:
$$\sum_{i=0}^{T-1}((p'_i-p_i-1) \mod \rho)$$

This measurement can be repeated as long as the device does not restart.

%% file: app-algorithm5.tex
\section{Analysis of RFC 6056's Algorithm 5}
\label{app:algorithm5}
RFC 6056's Algorithm 5 increments a global counter by a random value between 1 and $N$, where $N$ is configurable. In order to avoid connection id reuse, RFC 6056's Algorithm 5 should ensure that the counter does not wrap around in less than $2\cdot \mathit{MSL}$ seconds, where $\mathit{MSL}$ is the {\em server} TCP stack parameter. The original TCP RFC 793 sets $\mathit{MSL}=120$. The default value for Windows servers (the registry value {\tt TcpTimedWaitDelay}) is 60 seconds \cite{windows-msl}. In Linux, $\mathit{MSL}=30$ seconds (evident from the kernel constant {\tt TCP\_TIMEWAIT\_LEN}$=2\cdot MSL=60$ \cite{linux-msl}). The average progress per TCP port in RFC 6056's Algorithm 5 is $\frac{N+1}{2}$, therefore in order not to wrap around before $2\cdot \mathit{MSL}$ seconds have elapsed, the following condition is necessary (but not sufficient):
$$2\cdot\mathit{MSL}\cdot\frac{N+1}{2}\cdot r<R$$
Where $R$ is the port range (for Linux/Android, $R=60999-32768+1=28232$), and $r$ is the outbound TCP connection rate. The above upper bound for $N$ is not tight, because it assumes that each connection is short lived, i.e. terminated very shortly after it is established. If a connection is long lived, then its TIME\_WAIT phase is achieved after even more ports are consumed, thus lowering the bound for $N$.
In an anecdotal test with a Linux laptop running Ubuntu 20.04.3 and Chrome 96.0.4664.110, we opened several tabs for media-rich websites and got 737 TCP connections in a 64.5 seconds time interval, thus we measured $r=11.4$ connections/s. For Linux servers, this yields $N \leq 81$, and for Windows servers, this yields $N \leq 19$. As we noted above, this is a very loose upper bound for $N$. And it is quite possible that $r$ higher than 11.4 is common in some scenarios. But even setting $N=81$ yields low security since it reduces the entropy of the TCP source port by 8.5 bits (in the Linux server case), from $\log_2{28232}=14.8$ to $\log_2{81}=6.3$ (10.6 entropy bit reduction, to 4.2 bits in the Windows case). This seems unacceptable security-wise, and so Algorithm 5 fails to deliver a practical trade-off between security and functionality. 

Interestingly, RFC 6056 suggests $N=500$ without explanation how this value is obtained.

%% file: app-cloud-experiment.tex
\section{Cloud Providers Experiment}\label{sec:cloud-providers-experiment}
In this experiment, we deployed servers on multiple cloud providers across different regions and tested whether their networks rewrite TCP source ports. Our test uses a utility that contacts a reference server on a bound source port (random but known to the utility), to which the server replies with the observed source port. Our experiment shows that all 25 tested regions across 4 cloud providers (Azure, AWS, Google Cloud and Digital Ocean) do \emph{not} rewrite TCP source port. \autoref{tab:tested-cloud-providers-ipv4} summarizes the results.

\begin{table}[t]
\centering
\begin{threeparttable}

\begin{tabular}{l|l|l}
\toprule
Cloud provider                             & Region & \begin{tabular}{@{}l@{}}Preserve \\ Ports?\end{tabular}  \\ \midrule
\multirow{6}{*}{Azure} & Australia East & \checkmark \\
                       & Sweden Central & \checkmark \\
                       & Central US & \checkmark \\
                       & Central India & \checkmark \\
                       & UK West & \checkmark \\
                       & Korea Central & \checkmark \\
\midrule
\multirow{7}{*}{AWS}   & Oregon & \checkmark \\
                       & Milan & \checkmark \\
                       & Canada & \checkmark \\
                       & Sau Paulo & \checkmark \\
                       & Singapore & \checkmark \\
                       & Cape Town & \checkmark \\
                       & Hong Kong & \checkmark \\
\midrule
\multirow{7}{*}{Google Cloud}   & Las Vegas & \checkmark \\
                       & Santiago & \checkmark \\
                       & Madrid & \checkmark \\
                       & Taiwan & \checkmark \\
                       & Jakarta & \checkmark \\
                       & Melbourne & \checkmark \\
                       & Zurich & \checkmark \\
\midrule
\multirow{5}{*}{Digital Ocean} & Amsterdam & \checkmark \\
                       & Bangalore & \checkmark \\
                       & Singapore & \checkmark \\
                       & New York & \checkmark \\
                       & Toronto & \checkmark \\

\bottomrule
\end{tabular}

\end{threeparttable}
 \caption{
 Tested cloud providers and regions under IPv4
}
\label{tab:tested-cloud-providers-ipv4}
\end{table}